\documentclass[fleqn]{llncs}
\usepackage{version}
\usepackage{bis}

\title{On the Behaviours Produced by Instruction Sequences under
  Execution}
\author{J.A. Bergstra \and C.A. Middelburg}
\institute{Informatics Institute, Faculty of Science,
           University of Amsterdam, \\
           Science Park~904, 1098~XH Amsterdam, the Netherlands \\
           \email{J.A.Bergstra@uva.nl,C.A.Middelburg@uva.nl}}

\begin{document}

\maketitle

\begin{abstract}
We study several aspects of the behaviours produced by instruction
sequences under execution in the setting of the algebraic theory of
processes known as \ACP.
We use \ACP\ to describe the behaviours produced by instruction
sequences under execution and to describe two protocols implementing
these behaviours in the case where the processing of instructions takes
place remotely.
We also show that all finite-state behaviours considered in \ACP\ can
be produced by instruction sequences under execution.
\end{abstract}


\section{Introduction}
\label{sect-intro}

The concept of an instruction sequence is a very primitive concept in
computing.
It has always played a central role in computing because of the fact
that execution of instruction sequences underlies virtually all past
and current generations of computers.
It happens that, given a precise definition of an appropriate notion
of an instruction sequence, many issues in computer science can be
clearly explained in terms of instruction sequences.
A simple yet interesting example is that a program can simply be
defined as a text that denotes an instruction sequence.
Such a definition corresponds to an empirical perspective found among
practitioners.

In theoretical computer science, the meaning of programs usually plays
a prominent part in the explanation of many issues concerning programs.
Moreover, what is taken for the meaning of programs is mathematical by
nature.
On the other hand, it is customary that practitioners do not fall back
on the mathematical meaning of programs in case explanation of issues
concerning programs is needed.
They phrase their explanations from an empirical perspective.
An empirical perspective that we consider appealing is the perspective
that a program is in essence an instruction sequence and an instruction
sequence under execution produces a behaviour that is controlled by its
execution environment in the sense that each step performed actuates
the processing of an instruction by the execution environment and a
reply returned at completion of the processing determines how the
behaviour proceeds.

This paper concerns the behaviours produced by instruction sequences
under execution as such and two issues relating to the behaviours
produced by instruction sequences under execution, namely the issue of
implementing these behaviours in the case where the processing of
instructions takes place remotely and the issue of the extent to which
the behaviours considered in process algebra can be produced by
instruction sequences under execution.

Remote instruction processing means that a stream of instructions to be
processed arises at one place and the processing of that stream of
instructions is handled at another place.
This phenomenon is increasingly encountered.
It is found if loading the instruction sequence to be executed as a
whole is impracticable.
For instance, the storage capacity of the execution unit is too small
or the execution unit is too far away.
Remote instruction processing requires special attention because the
transmission time of the messages involved in remote instruction
processing makes it hard to keep the execution unit busy without
intermission.

In the literature on computer architecture, hardly anything can be
found that contributes to a sound understanding of the phenomenon of
remote instruction processing.
As a first step towards such an understanding, we give rigorous
descriptions of two protocols for remote instruction processing at a
level of abstraction that captures the underlying essence of the
protocols.
One protocol is very simple, but makes no effort keep the execution
unit busy without intermission.
The other protocol is more complex and is directed towards keeping the
execution unit busy without intermission.
It is reminiscent of an instruction pre-fetching mechanism as found in
pipelined processors (see e.g.~\cite{HP03a}), but its range of
application is not restricted to pipelined instruction processing.

The work presented in this paper belongs to a line of research which
started with an attempt to approach the semantics of programming
languages from the perspective mentioned above.
The first published paper on this approach is~\cite{BL00a}.
That paper is superseded by~\cite{BL02a} with regard to the groundwork
for the approach: program algebra, an algebraic theory of single-pass
instruction sequences, and basic thread algebra, an algebraic theory of
mathematical objects that represent in a direct way the behaviours
produced by instruction sequences under execution.%
\footnote
{In~\cite{BL02a}, basic thread algebra is introduced under the name
 basic polarized process algebra.
}
The main advantages of the approach are that it does not require a lot
of mathematical background and that it is more appealing to
practitioners than the main approaches to programming language
semantics: the operational approach, the denotational approach and the
axiomatic approach.
For an overview of these approaches, see e.g.~\cite{Mos06a}.

The work presented in this paper is based on the instruction sequences
considered in program algebra and the representation of the behaviours
produced by instruction sequences under execution considered in basic
thread algebra.
It is rather awkward to describe and analyse the behaviours of this
kind using algebraic theories of processes such as
\ACP~\cite{BW90,BK84b}, CCS~\cite{HM85,Mil89} and
CSP~\cite{BHR84,Hoa85}.
However, the objects considered in basic thread algebra can be viewed
as representations of processes as considered in process algebra.
This allows for the protocols for remote instruction processing to be
described using \ACP\ or rather \ACPt, an extension of \ACP\ which
supports abstraction from internal actions.

Process algebra is an area of the study of concurrency which is
considered relevant to computer science, as is witnesses by the extent
of the work on algebraic theories of processes such as \ACP, CCS and
CSP in theoretical computer science.
This strongly hints that there are programmed systems whose behaviours
can be taken for processes as considered in process algebra.
Therefore, it is interesting to know to which extent the behaviours
considered in process algebra can be produced by programs under
execution, starting from the perception of a program as an instruction
sequence.
In this paper, we will show that, by apposite choice of instructions,
all finite-state processes can be produced by instruction sequences
(provided that the cluster fair abstraction rule, see e.g.\
Section~5.6 of ~\cite{Fok00}, is valid).

The instruction sequences considered in program algebra are single-pass
instruction sequences, i.e.\ finite or infinite sequences of
instructions of which each instruction is executed at most once and can
be dropped after it has been executed or jumped over.
Program algebra does not provide a notation for programs that is
intended for actual programming: programs written in an assembly
language are finite instruction sequences for which single-pass
execution is usually not possible.
We will also show that all finite-state processes can as well be
produced by programs written in a program notation which is close to
existing assembly languages.

Instruction sequences under execution may make use of services provided
by their execution environment such as counters, stacks and Turing
tapes.
The use operators added to basic thread algebra in e.g.~\cite{BM07g}
can be used to describe the behaviours produced by instruction
sequences under execution that make use of services.
Interesting is that instruction sequences under execution that make use
of services may produce infinite-state processes.
On that account, we will make precise what processes are produced by
instruction sequences under execution that make use of services
provided by their execution environment.

As a continuation of the work on a new approach to programming language
semantics mentioned above, the notion of an instruction sequence was
subjected to systematic and precise analysis using the groundwork laid
earlier.
This led among other things to expressiveness results about the
instruction sequences considered and variations of the instruction
sequences considered (see e.g.~\cite{BM07g,BM08b,BM08h,BP09a,PZ06a}).
Instruction sequences are under discussion for many years in diverse
work on computer architecture, as witnessed by
e.g.~\cite{Bak91a,BH97a,HJP82a,Lun77a,NH97a,OH00a,PD80a,TW07a,XT96a},
but the notion of an instruction sequence has never been subjected to
any precise analysis before.
As another continuation of the work on a new approach to programming
language semantics mentioned above, selected issues relating to
well-known subjects from the theory of computation and the area of
computer architecture were rigorously investigated thinking in terms of
instruction sequences (see e.g.~\cite{BM09c,BM08g,BM07c,BM09i,BM09k}).
The subjects from the theory of computation, namely the halting problem
and non-uniform computational complexity, are usually investigated
thinking in terms of a common model of computation such as Turing
machines and Boolean circuits (see e.g.~\cite{AB09a,Her65a,Sip06a}).
The subjects from the area of computer architecture, namely instruction
sequence performance, instruction set architectures and remote
instruction processing, are usually not investigated in a rigorous way
at all.
The general aim of the work in both continuations mentioned is to bring
instruction sequences as a theme in computer science better into the
picture.
The work presented in this paper forms a part of the last mentioned
continuation.

This paper is organized as follows.
The body of the paper consists of three parts.
The first part (Sections~\ref{sect-PGA}--\ref{sect-process-extr})
concerns the behaviours produced by instruction sequences under
execution as such and includes surveys of program algebra, basic thread
algebra and the algebraic theory of processes known as ACP.
The second part
(Sections~\ref{sect-simple-protocol}--\ref{sect-adaptations}) concerns
the issue of implementing these behaviours in the case where the
processing of instructions takes place remotely and includes rigorous
descriptions of two protocols for remote instruction processing.
The third part (Sections~\ref{sect-pci}--\ref{sect-PGLD-BR}) concerns
the issue of the extent to which the behaviours considered in process
algebra can be produced by instruction sequences under execution and
includes the result that, by apposite choice of instructions, all
finite-state processes can be produced by instruction sequences.

This paper consolidates material from~\cite{BM08i,BM09e,BM09c}.

\section{Program Algebra}
\label{sect-PGA}

In this section, we review \PGA\ (ProGram Algebra).
The starting-point of program algebra is the perception of a program as
a single-pass instruction sequence.
The concepts underlying the primitives of program algebra are common in
programming, but the particular form of the primitives is not common.
The predominant concern in the design of program algebra has been to
achieve simple syntax and semantics, while maintaining the expressive
power of arbitrary finite control.

In \PGA, it is assumed that a fixed but arbitrary set $\BInstr$
of \emph{basic instructions} has been given.
The intuition is that the execution of a basic instruction may modify a
state and produces a reply at its completion.
The possible replies are the Boolean values $\True$ and $\False$.

\PGA\ has the following \emph{primitive instructions}:
\begin{itemize}
\item
for each $a \in \BInstr$, a \emph{plain basic instruction} $a$;
\item
for each $a \in \BInstr$, a \emph{positive test instruction} $\ptst{a}$;
\item
for each $a \in \BInstr$, a \emph{negative test instruction} $\ntst{a}$;
\item
for each $l \in \Nat$, a \emph{forward jump instruction} $\fjmp{l}$;
\item
a \emph{termination instruction} $\halt$.
\end{itemize}
We write $\PInstr$ for the set of all primitive instructions of \PGA.
On execution of an instruction sequence, these primitive instructions
have the following effects:
\begin{itemize}
\item
the effect of a positive test instruction $\ptst{a}$ is that basic
instruction $a$ is executed and execution proceeds with the next
primitive instruction if $\True$ is produced and otherwise the next
primitive instruction is skipped and execution proceeds with the
primitive instruction following the skipped one --- if there is no
primitive instructions to proceed with, inaction occurs;
\item
the effect of a negative test instruction $\ntst{a}$ is the same as the
effect of $\ptst{a}$, but with the role of the value produced reversed;
\item
the effect of a plain basic instruction $a$ is the same as the effect of
$\ptst{a}$, but execution always proceeds as if $\True$ is produced;
\item
the effect of a forward jump instruction $\fjmp{l}$ is that execution
proceeds with the $l$-th next instruction of the program concerned ---
if $l$ equals $0$ or there is no primitive instructions to proceed with,
inaction occurs;
\item
the effect of the termination instruction $\halt$ is that execution
terminates.
\end{itemize}

\PGA\ has the following constants and operators:
\begin{itemize}
\item
for each $u \in \PInstr$, an \emph{instruction} constant $u$\,;
\item
the binary \emph{concatenation} operator ${}\conc{}$\,;
\item
the unary \emph{repetition} operator $\,\rep$\,.
\end{itemize}
We assume that there is a countably infinite set of variables which
includes $x,y,z$.
Terms are built as usual.
We use infix notation for concatenation and postfix notation for
repetition.

A closed \PGA\ term is considered to denote a non-empty, finite or
eventually periodic infinite sequence of primitive instructions.%
\footnote
{An eventually periodic infinite sequence is an infinite sequence with
 only finitely many distinct suffixes.}
The instruction sequence denoted by a closed term of the form
$t \conc t'$ is the instruction sequence denoted by $t$ concatenated
with the instruction sequence denoted by $t'$.
The instruction sequence denoted by a closed term of the form $t\rep$
is the instruction sequence denoted by $t$ concatenated infinitely
many times with itself.
Some simple examples of closed \PGA\ terms are
\begin{ldispl}
a \conc b \conc c\;,
\qquad
\ptst{a} \conc \fjmp{2} \conc \fjmp{3} \conc b \conc \halt\;,
\qquad
a \conc (b \conc c)\rep\;.
\end{ldispl}%
On execution of the instruction sequence denoted by the first term, the
basic instructions $a$, $b$ and $c$ are executed in that order and
after that inaction occurs.
On execution of the instruction sequence denoted by the second term,
the basic instruction $a$ is executed first, if the execution of $a$
produces the reply $\True$, the basic instruction $b$ is executed next
and after that execution terminates, and if the execution of $a$
produces the reply $\False$, inaction occurs.
On execution of the instruction sequence denoted by the third term, the
basic instruction $a$ is executed first, and after that the basic
instructions $b$ and $c$ are executed in that order repeatedly forever.

Closed \PGA\ terms are considered equal if they represent the same
instruction sequence.
The axioms for instruction sequence equivalence are given in
Table~\ref{axioms-PGA}.%
\begin{table}[!t]
\caption{Axioms of \PGA}
\label{axioms-PGA}
\begin{eqntbl}
\begin{axcol}
(x \conc y) \conc z = x \conc (y \conc z)              & \axiom{PGA1} \\
(x^n)\rep = x\rep                                      & \axiom{PGA2} \\
x\rep \conc y = x\rep                                  & \axiom{PGA3} \\
(x \conc y)\rep = x \conc (y \conc x)\rep              & \axiom{PGA4}
\end{axcol}
\end{eqntbl}
\end{table}
In this table, $n$ stands for an arbitrary positive natural number.
The term $t^n$, where $t$ is a \PGA\ term, is defined by induction on
$n$ as follows: $t^1 = t$ and $t^{n+1} = t \conc t^n$.
The \emph{unfolding} equation $x\rep = x \conc x\rep$ is derived as
follows:
\begin{ldispl}
\begin{deriv}
x\rep & = & (x \conc x)\rep         & \axiom{PGA2} \\
      & = & x \conc (x \conc x)\rep & \axiom{PGA4} \\
      & = & x \conc x\rep           & \axiom{PGA2}\;.
\end{deriv}
\end{ldispl}%
Each closed \PGA\ term is derivably equal to a term in
\emph{canonical form}, i.e.\ a term of the form $t$ or
$t \conc {t'}\rep$, where $t$ and $t'$ are closed \PGA\ terms in which
the repetition operator does not occur.
For example:
\begin{ldispl}
(a \conc b)\rep \conc c \conc \halt =
a \conc (b \conc a)\rep\;,
\\
\ptst{a} \conc (\fjmp{4} \conc b \conc
                  (\ntst{c} \conc \fjmp{5} \conc \halt)\rep)\rep
=
\ptst{a} \conc \fjmp{4} \conc b \conc
(\ntst{c} \conc \fjmp{5} \conc \halt)\rep\;.
\end{ldispl}%

The initial models of \PGA\ are considered its standard models.
Henceforth, we restrict ourselves to the initial model $\Ipga$ of \PGA\
in which:
\begin{itemize}
\item
the domain is the set of all non-empty, finite and eventually periodic
infinite sequences over the set $\PInstr$ of primitive instructions;
\item
the operation associated with ${} \conc {}$ is concatenation;
\item
the operation associated with ${}\rep$ is the operation ${}\srep$
defined as follows:
\begin{itemize}
\item
if $F$ is a finite sequence over $\PInstr$, then $F\srep$ is the unique
eventually periodic infinite sequence $F'$ such that $F$ concatenated $n$
times with itself is a proper prefix of $F'$ for each $n \in \Nat$;
\item
if $F$ is an eventually periodic infinite sequence over $\PInstr$, then
$F\srep$ is $F$.
\end{itemize}
\end{itemize}

In the sequel, we use the term \emph{instruction sequence} for the
elements of the domain of $\Ipga$, and we denote the interpretations of
the constants and operators of \PGA\ in $\Ipga$ by the constants and
operators themselves.
$\Ipga$ is loosely called \emph{the} initial model of \PGA\ because all
initial models of \PGA\ are isomorphic, i.e.\ there exist bijective
homomorphism between them (see e.g.~\cite{ST99a,Wir90a}).

\section{Basic Thread Algebra}
\label{sect-BTA}

In this section, we review \BTA\ (Basic Thread Algebra).
\BTA\ is an algebraic theory of mathematical objects that represent in
a direct way the behaviours produced by instruction sequences under
execution.
The objects concerned are called threads.

In \BTA, it is assumed that a fixed but arbitrary set $\BAct$ of
\emph{basic actions}, with $\Tau \notin \BAct$, has been given.
Besides, $\Tau$ is a special basic action.
We write $\BActTau$ for $\BAct \union \set{\Tau}$.
A thread performs basic actions in a sequential fashion.
Upon each basic action performed, a reply from an execution environment
determines how it proceeds.
The possible replies are the Boolean values $\True$ and~$\False$.
Performing $\Tau$, which is considered performing an internal action,
always leads to the reply $\True$.

Although \BTA\ is one-sorted, we make this sort explicit.
The reason for this is that we will extend \BTA\ with an additional sort
in Section~\ref{sect-services}.

\BTA\ has one sort: the sort $\Thr$ of \emph{threads}.
To build terms of sort $\Thr$, it has the following constants and
operators:
\begin{itemize}
\item
the \emph{inaction} constant $\const{\DeadEnd}{\Thr}$;
\item
the \emph{termination} constant $\const{\Stop}{\Thr}$;
\item
for each $a \in \BActTau$, the binary \emph{postconditional composition}
operator $\funct{\pccop{a}}{\Thr \x \Thr}{\Thr}$.
\end{itemize}
We assume that there are infinitely many variables of sort $\Thr$,
including $x,y,z$.
Terms of sort $\Thr$ are built as usual.
We use infix notation for the postconditional composition operators.
We introduce \emph{basic action prefixing} as an abbreviation:
$a \bapf t$, where $a \in \BActTau$ and $t$ is a term of sort $\Thr$,
abbreviates $\pcc{t}{a}{t}$.

The thread denoted by a closed term of the form $\pcc{t}{a}{t'}$ will
first perform $a$, and then proceed as the thread denoted by $t$ if the
reply from the execution environment is $\True$ and proceed as the
thread denoted by $t'$ if the reply from the execution environment is
$\False$.
The threads denoted by $\DeadEnd$ and $\Stop$ will become inactive and
terminate, respectively.
Some simple examples of closed \BTA\ terms are
\begin{ldispl}
a \bapf (\pcc{\Stop}{b}{\DeadEnd})\;, \qquad
\pcc{(b \bapf \Stop)}{a}{\DeadEnd}\;.
\end{ldispl}%
The first term denotes the thread that first performs basic action $a$,
next performs basic action $b$, if the reply from the execution
environment on performing $b$ is $\True$, after that terminates, and if
the reply from the execution environment on performing $b$ is $\False$,
after that becomes inactive.
The second term denotes the thread that first performs basic action
$a$, if the reply from the execution environment on performing $a$ is
$\True$, next performs the basic action $b$ and after that terminates,
and if the reply from the execution environment on performing $a$ is
$\False$, next becomes inactive.

\BTA\ has only one axiom.
This axiom is given in Table~\ref{axioms-BTA}.%
\begin{table}[!tb]
\caption{Axiom of \BTA}
\label{axioms-BTA}
\begin{eqntbl}
\begin{axcol}
\pcc{x}{\Tau}{y} = \pcc{x}{\Tau}{x}                    & \axiom{T1}
\end{axcol}
\end{eqntbl}
\end{table}
Using the abbreviation introduced above, axiom T1 can be written as
follows: $\pcc{x}{\Tau}{y} = \Tau \bapf x$.

Notice that each closed \BTA\ term denotes a thread that will become
inactive or terminate after it has performed finitely many actions.
Infinite threads can be described by guarded recursion.

A \emph{guarded recursive specification} over \BTA\ is a set of
recursion equations $E = \set{X = t_X \where X \in V}$, where $V$ is a
set of variables of sort $\Thr$ and each $t_X$ is a \BTA\ term of the
form $\DeadEnd$, $\Stop$ or $\pcc{t}{a}{t'}$ with $t$ and $t'$ that
contain only variables from $V$.
We write $\vars(E)$ for the set of all variables that occur in $E$.
We are only interested in models of \BTA\ in which guarded recursive
specifications have unique solutions, such as the projective limit model
of \BTA\ presented in~\cite{BB03a}.

A simple example of a guarded recursive specification is the one
consisting of following two equations:
\begin{ldispl}
x = \pcc{x}{a}{y}\;, \qquad
y = \pcc{y}{b}{\Stop}\;.
\end{ldispl}%
The $x$-component of the solution of this guarded recursive
specification is the thread that first performs basic action $a$
repeatedly until the reply from the execution environment on performing
$a$ is $\False$, next performs basic action $b$ repeatedly until the
reply from the execution environment on performing $b$ is $\False$, and
after that terminates.

For each guarded recursive specification $E$ and each $X \in \vars(E)$,
we introduce a constant $\rec{X}{E}$ of sort $\Thr$ standing for the
$X$-component of the unique solution of $E$.
We write $\rec{t_X}{E}$ for $t_X$ with, for all $Y \in \vars(E)$, all
occurrences of $Y$ in $t_X$ replaced by $\rec{Y}{E}$.
The axioms for the constants for the components of the solutions of
guarded recursive specifications are RDP (Recursive Definition
Principle) and RSP (Recursive Specification Principle), which are given
in Table~\ref{axioms-BTA-recursion}.%
\begin{table}[!t]
\caption{RDP, RSP and AIP}
\label{axioms-BTA-recursion}
\begin{eqntbl}
\begin{saxcol}
\rec{X}{E} = \rec{t_X}{E}        & \mif X = t_X \in E  & \axiom{RDP} \\
E \Limpl X = \rec{X}{E}          & \mif X \in \vars(E) & \axiom{RSP} \\
{}                                                                   \\
\multicolumn{2}{@{}l@{}}
 {\LAND{n \geq 0} \proj{n}{x} = \proj{n}{y} \Limpl x = y}
                                                       & \axiom{AIP}
\end{saxcol}
\qquad
\begin{axcol}
\proj{0}{x} = \DeadEnd                                  & \axiom{P0} \\
\proj{n+1}{\Stop} = \Stop                               & \axiom{P1} \\
\proj{n+1}{\DeadEnd} = \DeadEnd                         & \axiom{P2} \\
\proj{n+1}{\pcc{x}{a}{y}} =
                      \pcc{\proj{n}{x}}{a}{\proj{n}{y}} & \axiom{P3}
\end{axcol}
\end{eqntbl}
\end{table}
RDP and RSP are actually axiom schemas in which $X$ stands for an
arbitrary variable, $t_X$ stands for an arbitrary \BTA\ term, and $E$
stands for an arbitrary guarded recursive specification over \BTA.
Side conditions are added to restrict what $X$, $t_X$ and $E$ stand for.
The equations $\rec{X}{E} = \rec{t_X}{E}$ for a fixed $E$ express that
the constants $\rec{X}{E}$ make up a solution of $E$.
The conditional equations $E \Limpl X = \rec{X}{E}$ express that this
solution is the only one.

RDP and RSP are means to prove closed terms that denote the same
infinite thread equal.
We introduce AIP (Approximation Induction Principle) as an additional
means to prove closed terms that denote the same infinite thread equal.
AIP is based on the view that two threads are identical if their
approximations up to any finite depth are identical.
The approximation up to depth $n$ of a thread is obtained by cutting it
off after it has performed $n$ actions.
AIP is also given in Table~\ref{axioms-BTA-recursion}.
Here, approximation up to depth $n$ is phrased in terms of the
unary \emph{projection} operator $\funct{\projop{n}}{\Thr}{\Thr}$.
The axioms for the projection operators are axioms P0--P3 in
Table~\ref{axioms-BTA-recursion}.
P1--P3 are actually axiom schemas in which $a$ stands for arbitrary
basic action and $n$ stands for an arbitrary natural number.

We write \BTArec\ for \BTA\ extended with the constants for the
components of the solutions of guarded recursive specifications, the
projection operators and the axioms RDP, RSP, AIP and P0--P3.

The minimal models of \BTArec\ are considered its standard models.%
\footnote
{A minimal model of an algebraic theory is a model of which no proper
 subalgebra is a model as well.
}
Recall that a model of an algebraic theory is minimal iff all elements
of the domains associated with the sorts of the theory can be denoted
by closed terms.
Henceforth, we restrict ourselves to the minimal models of \BTArec.
We assume that a minimal model $\Mbta$ of \BTArec\ has been given.

In the sequel, we use the term \emph{thread} for the elements of the
domain of $\Mbta$, and we denote the interpretations of constants and
operators in $\Mbta$ by the constants and operators themselves.

Let $T$ be a thread.
Then the set of \emph{states} or \emph{residual threads} of $T$,
written $\Res(T)$, is inductively defined as follows:
\begin{itemize}
\item
$T \in \Res(T)$;
\item
if $\pcc{T'}{a}{T''} \in \Res(T)$, then $T' \in \Res(T)$ and
$T'' \in \Res(T)$.
\end{itemize}

Let $T$ be a thread and let $\BAct' \subseteq \BActTau$.
Then $T$ is \emph{regular over} $\BAct'$ if the following conditions are
satisfied:
\begin{itemize}
\item
$\Res(T)$ is finite;
\item
for all $T',T'' \in \Res(T)$ and $a \in \BActTau$,
$\pcc{T'}{a}{T''} \in \Res(T)$ implies $a \in \BAct'$.
\end{itemize}
We say that $T$ is \emph{regular} if $T$ is regular over $\BActTau$.

For example, the $x$-component of the solution of the guarded recursive
specification consisting of the following two equations:
\begin{ldispl}
x = a \bapf y\;, \qquad
y = \pcc{(c \bapf y)}{b}{(\pcc{x}{d}{\Stop})}
\end{ldispl}%
has five states and is regular over any $\BAct' \subseteq \BActTau$ for
which $\set{a,b,c,d} \subseteq \BAct'$.

We will make use of the fact that being a regular thread coincides with
being a component of the solution of a finite guarded recursive
specification in which the right-hand sides of the recursion equations
are of a restricted form.

A \emph{linear recursive specification} over \BTA\ is a guarded
recursive specification $E = \set{X = t_X \where X \in V}$ over
\BTA, where each $t_X$ is a term of the form $\DeadEnd$, $\Stop$ or
$\pcc{Y}{a}{Z}$ with $Y,Z \in V$.
\begin{proposition}
\label{prop-lin-rec-thread}
Let $T$ be a thread and let $\BAct' \subseteq \BActTau$.
Then $T$ is regular over $\BAct'$ iff there exists a finite linear
recursive specification $E$ over \BTA\ in which only basic actions from
$\BAct'$ occur such that $T$ is a component of the solution of $E$.
\end{proposition}
\begin{proof}
The implication from left to right is proved as follows.
Because $T$ is regular, $\Res(T)$ is finite.
Hence, there are finitely many threads $T_1$, \ldots, $T_n$, with
$T = T_1$, such that $\Res(T) = \set{T_1,\ldots,T_n}$.
Now $T$ is the $x_1$-component of the solution of the linear recursive
specification consisting of the following equations:
\begin{ldispl}
\begin{gceqns}
x_i =
\left\{
\begin{array}[c]{@{}l@{\;\;}l@{}}
\Stop    & \mif T_i = \Stop \\
\DeadEnd & \mif T_i = \DeadEnd \\
\pcc{x_j}{a}{x_k} & \mif T_i = \pcc{T_j}{a}{T_k}
\end{array}
\right.
& \mathrm{for\;all}\; i \in [1,n]\;.
\end{gceqns}
\end{ldispl}%
Because $T$ is regular over $\BAct'$, only basic actions from $\BAct'$
occur in the linear recursive specification constructed in this way.

The implication from right to left is proved as follows.
Thread $T$ is a component of the unique solution of a finite linear
specification in which only basic actions from $\BAct'$ occur.
This means that there are finitely many threads $T_1$, \ldots, $T_n$,
with $T = T_1$, such that for every $i \in [1,n]$,
$T_i = \Stop$, $T_i = \DeadEnd$ or
$T_i = \pcc{T_j}{a}{T_k}$ for some $j,k \in [1,n]$ and
$a \in \BAct'$.
Consequently, $T' \in \Res(T)$ iff $T' = T_i$ for some $i \in [1,n]$
and moreover $\pcc{T'}{a}{T''} \in \Res(T)$ only if $a \in \BAct'$.
Hence, $\Res(T)$ is finite and $T$ is regular over $\BAct'$.
\qed
\end{proof}

\begin{remark}
A structural operational semantics of \BTArec\ and a bisimulation
equivalence based on it can be found in e.g.~\cite{BM05f}.
The quotient algebra of the algebra of closed terms of \BTArec\ by
this bisimulation equivalence is one of the minimal models of \BTArec.
\end{remark}

\section{Thread Extraction}
\label{sect-thread-extr}

In this short section, we use \BTArec\ to make mathematically precise
which threads are produced by instruction sequences under execution.

For that purpose, $\BAct$ is taken such that $\BAct \supseteq \BInstr$
is satisfied.

The \emph{thread extraction} operation $\textr{\ph}$ assigns a thread to
each instruction sequence.
The thread extraction operation is defined by the equations given in
Table~\ref{axioms-thread-extr} (for $a \in \BInstr$, $l \in \Nat$, and
$u \in \PInstr$)%
\begin{table}[!t]
\caption{Defining equations for thread extraction operation}
\label{axioms-thread-extr}
\begin{eqntbl}
\begin{eqncol}
\textr{a} = a \bapf \DeadEnd \\
\textr{a \conc F} = a \bapf \textr{F} \\
\textr{\ptst{a}} = a \bapf \DeadEnd \\
\textr{\ptst{a} \conc F} =
\pcc{\textr{F}}{a}{\textr{\fjmp{2} \conc F}} \\
\textr{\ntst{a}} = a \bapf \DeadEnd \\
\textr{\ntst{a} \conc F} =
\pcc{\textr{\fjmp{2} \conc F}}{a}{\textr{F}}
\end{eqncol}
\qquad
\begin{eqncol}
\textr{\fjmp{l}} = \DeadEnd \\
\textr{\fjmp{0} \conc F} = \DeadEnd \\
\textr{\fjmp{1} \conc F} = \textr{F} \\
\textr{\fjmp{l+2} \conc u} = \DeadEnd \\
\textr{\fjmp{l+2} \conc u \conc F} = \textr{\fjmp{l+1} \conc F} \\
\textr{\halt} = \Stop \\
\textr{\halt \conc F} = \Stop
\end{eqncol}
\end{eqntbl}
\end{table}
and the rule that $\textr{\fjmp{l} \conc F} = \DeadEnd$ if $\fjmp{l}$ is
the beginning of an infinite jump chain.
This rule is formalized in e.g.~\cite{BM07g}.

Let $F$ be an instruction sequence and $T$ be a thread.
Then we say that $F$ \emph{produces} $T$ if $\textr{F} = T$.
For example,
\begin{ldispl}
\begin{aeqns}
a \conc b \conc c
& \quad \mathrm{produces} \quad &
a \bapf b \bapf c \bapf \DeadEnd\;,
\\
\ptst{a} \conc \fjmp{2} \conc \fjmp{3} \conc b \conc \halt
& \quad \mathrm{produces} \quad &
\pcc{(b \bapf \Stop)}{a}{\DeadEnd}\;,
\\
\ptst{a} \conc \ntst{b} \conc c \conc \halt
& \quad \mathrm{produces} \quad &
\pcc{(\pcc{\Stop}{b}{(c \bapf \Stop)})}
    {a}{(c \bapf \Stop)}\;,
\\
\ptst{a} \conc \fjmp{2} \conc
(b \conc \fjmp{2} \conc c \conc \fjmp{2})\rep
& \quad \mathrm{produces} \quad &
\pcc{\DeadEnd}{a}{(b \bapf \DeadEnd)}\;.
\end{aeqns}
\end{ldispl}%
In the case of instruction sequences that are not finite, the produced
threads can be described as components of the solution of a guarded
recursive specification.
For example, the infinite instruction sequence
\begin{ldispl}
(a \conc \ptst{b})\rep
\end{ldispl}%
produces the $x$-component of the solution of the guarded recursive
specification consisting of following two equations:
\begin{ldispl}
x = a \bapf y\;, \qquad
y = \pcc{x}{b}{y}
\end{ldispl}%
and the infinite instruction sequence
\begin{ldispl}
a \conc
(\ptst{b} \conc \fjmp{2} \conc \fjmp{3} \conc c \conc
 \fjmp{4} \conc \ntst{d} \conc \halt \conc a)\rep
\end{ldispl}%
produces the $x$-component of the solution of the guarded recursive
specification consisting of following two equations:
\begin{ldispl}
x = a \bapf y\;, \qquad
y = \pcc{(c \bapf y)}{b}{(\pcc{x}{d}{\Stop})}\;.
\end{ldispl}%

\section{Algebra of Communicating Processes}
\label{sect-ACP}

In this section, we review \ACPt\ (Algebra of Communicating Processes
with abstraction).
This algebraic theory of processes will among other things be used to
make precise what processes are produced by the threads denoted by
closed terms of \BTArec.
For a comprehensive overview of \ACPt, the reader is referred
to~\cite{BW90,Fok00}.

In \ACPt, it is assumed that a fixed but arbitrary set $\Act$ of
\emph{atomic actions}, with $\tau,\dead \notin \Act$, and a fixed but
arbitrary commutative and associative function
$\funct{\commm}{\Act \union \set{\tau} \x \Act \union \set{\tau}}
               {\Act \union \set{\dead}}$,
with $\tau \commm e = \dead$ for all $e \in \Act \union \set{\tau}$,
have been given.
The function $\commm$ is regarded to give the result of synchronously
performing any two atomic actions for which this is possible, and to
give $\dead$ otherwise.
In \ACPt, $\tau$ is a special atomic action, called the silent step.
The act of performing the silent step is considered unobservable.
Because it would otherwise be observable, the silent step is considered
an atomic action that cannot be performed synchronously with other
atomic actions.
We write $\Actt$ for $\Act \union \set{\tau}$.

\ACPt\ has the following constants and operators:
\begin{itemize}
\item
for each $e \in \Act$, the \emph{atomic action} constant $e$\,;
\item
the \emph{silent step} constant $\tau$\,;
\item
the \emph{inaction} constant $\dead$\,;
\item
the binary \emph{alternative composition} operator ${}\altc{}$\,;
\item
the binary \emph{sequential composition} operator ${}\seqc{}$\,;
\item
the binary \emph{parallel composition} operator ${}\parc{}$\,;
\item
the binary \emph{left merge} operator ${}\leftm{}$\,;
\item
the binary \emph{communication merge} operator ${}\commm{}$\,;
\item
for each $H \subseteq \Act$, the unary \emph{encapsulation} operator
$\encap{H}$\,;
\item
for each $I \subseteq \Act$, the unary \emph{abstraction} operator
$\abstr{I}$\,.
\end{itemize}
We assume that there are infinitely many variables, including $x,y,z$.
Terms are built as usual.
We use infix notation for the binary operators.
The precedence conventions used with respect to the operators of \ACPt\
are as follows: $\altc$ binds weaker than all others, $\seqc$ binds
stronger than all others, and the remaining operators bind equally
strong.

Let $t$ and $t'$ be closed \ACPt\ terms, $e \in \Act$, and
$H,I \subseteq \Act$.
Intuitively, the constants and operators to build \ACPt\ terms can be
explained as follows:
\begin{itemize}
\item
the process denoted by $e$ first performs atomic action $e$ and next
terminates successfully;
\item
the process denoted by $\tau$ performs an unobservable atomic action and
next terminates successfully;
\item
the process denoted by $\dead$ can neither perform an atomic action nor
terminate successfully;
\item
the process denoted by $t \altc t'$ behaves either as the process denoted
by $t$ or as the process denoted by $t'$, but not both;
\item
the process denoted by $t \seqc t'$ first behaves as the process denoted
by $t$ and on successful termination of that process it next behaves as
the process denoted by $t'$;
\item
the process denoted by $t \parc t'$ behaves as the process that proceeds
with the processes denoted by $t$ and $t'$ in parallel;
\item
the process denoted by $t \leftm t'$ behaves the same as the process
denoted by $t \parc t'$, except that it starts with performing an atomic
action of the process denoted by $t$;
\item
the process denoted by $t \commm t'$ behaves the same as the process
denoted by $t \parc t'$, except that it starts with performing an atomic
action of the process denoted by $t$ and an atomic action of the process
denoted by $t'$ synchronously;
\item
the process denoted by $\encap{H}(t)$ behaves the same as the process
denoted by $t$, except that atomic actions from $H$ are blocked;
\item
the process denoted by $\abstr{I}(t)$ behaves the same as the process
denoted by $t$, except that atomic actions from $I$ are turned into
unobservable atomic actions.
\end{itemize}
The operators $\leftm$ and $\commm$ are of an auxiliary nature.
They are needed to axiomatize \ACPt.

The axioms of \ACPt\ are given in Table~\ref{axioms-ACPt}.
\begin{table}[!t]
\caption{Axioms of \ACPt}
\label{axioms-ACPt}
\begin{eqntbl}
\begin{axcol}
x \altc y = y \altc x                                  & \axiom{A1}  \\
(x \altc y) \altc z = x \altc (y \altc z)              & \axiom{A2}  \\
x \altc x = x                                          & \axiom{A3}  \\
(x \altc y) \seqc z = x \seqc z \altc y \seqc z        & \axiom{A4}  \\
(x \seqc y) \seqc z = x \seqc (y \seqc z)              & \axiom{A5}  \\
x \altc \dead = x                                      & \axiom{A6}  \\
\dead \seqc x = \dead                                  & \axiom{A7}  \\
{}                                                                   \\
x \parc y =
          x \leftm y \altc y \leftm x \altc x \commm y & \axiom{CM1} \\
a \leftm x = a \seqc x                                 & \axiom{CM2} \\
a \seqc x \leftm y = a \seqc (x \parc y)               & \axiom{CM3} \\
(x \altc y) \leftm z = x \leftm z \altc y \leftm z     & \axiom{CM4} \\
a \seqc x \commm b = (a \commm b) \seqc x              & \axiom{CM5} \\
a \commm b \seqc x = (a \commm b) \seqc x              & \axiom{CM6} \\
a \seqc x \commm b \seqc y =
                        (a \commm b) \seqc (x \parc y) & \axiom{CM7} \\
(x \altc y) \commm z = x \commm z \altc y \commm z     & \axiom{CM8} \\
x \commm (y \altc z) = x \commm y \altc x \commm z     & \axiom{CM9}
\end{axcol}
\qquad
\begin{axcol}
x \seqc \tau = x                                       & \axiom{B1}  \\
x \seqc (\tau \seqc (y \altc z) \altc y) = x \seqc (y \altc z)
                                                       & \axiom{B2}  \\
{}                                                                   \\
\encap{H}(a) = a                \hfill \mif a \notin H & \axiom{D1}  \\
\encap{H}(a) = \dead            \hfill \mif a \in H    & \axiom{D2}  \\
\encap{H}(x \altc y) = \encap{H}(x) \altc \encap{H}(y) & \axiom{D3}  \\
\encap{H}(x \seqc y) = \encap{H}(x) \seqc \encap{H}(y) & \axiom{D4}  \\
{}                                                                   \\
\abstr{I}(a) = a                \hfill \mif a \notin I & \axiom{TI1} \\
\abstr{I}(a) = \tau             \hfill \mif a \in I    & \axiom{TI2} \\
\abstr{I}(x \altc y) = \abstr{I}(x) \altc \abstr{I}(y) & \axiom{TI3} \\
\abstr{I}(x \seqc y) = \abstr{I}(x) \seqc \abstr{I}(y) & \axiom{TI4} \\
{}                                                                   \\
a \commm b = b \commm a                                & \axiom{C1}  \\
(a \commm b) \commm c = a \commm (b \commm c)          & \axiom{C2}  \\
\dead \commm a = \dead                                 & \axiom{C3}  \\
\tau \commm a = \dead                                  & \axiom{C4}
\end{axcol}
\end{eqntbl}
\end{table}
CM2--CM3, CM5--CM7, C1--C4, D1--D4 and TI1--TI4 are actually axiom
schemas in which $a$, $b$ and $c$ stand for arbitrary constants of
\ACPt, and $H$ and $I$ stand for arbitrary subsets of $\Act$.
\ACPt\ is extended with guarded recursion like \BTA.

A \emph{recursive specification} over \ACPt\ is a set of recursion
equations $E = \set{X = t_X \where X \in V}$, where $V$ is a set of
variables and each $t_X$ is an \ACPt\ term containing only variables
from $V$.
We write $\vars(E)$ for the set of all variables that occur in $E$.
Let $t$ be an \ACPt\ term without occurrences of abstraction operators
containing a variable $X$.
Then an occurrence of $X$ in $t$ is \emph{guarded} if $t$ has a subterm
of the form $e \seqc t'$ where $e \in \Act$ and $t'$ is a term
containing this occurrence of $X$.
Let $E$ be a recursive specification over \ACPt.
Then $E$ is a \emph{guarded recursive specification} if, in each
equation $X = t_X \in E$:
(i)~abstraction operators do not occur in $t_X$ and
(ii)~all occurrences of variables in $t_X$ are guarded or $t_X$ can be
rewritten to such a term using the axioms of \ACPt\ in either direction
and/or the equations in $E$ except the equation $X = t_X$ from left to
right.
We are only interested models of \ACPt\ in which guarded recursive
specifications have unique solutions, such as the models of \ACPt\
presented in~\cite{BW90}.

For each guarded recursive specification $E$ and each $X \in \vars(E)$,
we introduce a constant $\rec{X}{E}$ standing for the $X$-component of
the unique solution of $E$.
We write $\rec{t_X}{E}$ for $t_X$ with, for all $Y \in \vars(E)$, all
occurrences of $Y$ in $t_X$ replaced by $\rec{Y}{E}$.
The axioms for the constants for the components of the solutions of
guarded recursive specifications are RDP and RSP, which are given in
Table~\ref{axioms-ACPt-recursion}.%
\begin{table}[!t]
\caption{RDP, RSP and AIP}
\label{axioms-ACPt-recursion}
\begin{eqntbl}
\begin{saxcol}
\rec{X}{E} = \rec{t_X}{E}         & \mif X = t_X \in E  & \axiom{RDP} \\
E \Limpl X = \rec{X}{E}           & \mif X \in \vars(E) & \axiom{RSP} \\
{}                                                                    \\
\multicolumn{2}{@{}l@{}}
 {\LAND{n \geq 0} \proj{n}{x} = \proj{n}{y} \Limpl x = y}
                                                        & \axiom{AIP}
\end{saxcol}
\qquad
\begin{axcol}
\proj{0}{a} = \dead                                     & \axiom{PR1} \\
\proj{n+1}{a} = a                                       & \axiom{PR2} \\
\proj{0}{a \seqc x} = \dead                             & \axiom{PR3} \\
\proj{n+1}{a \seqc x} = a \seqc \proj{n}{x}             & \axiom{PR4} \\
\proj{n}{x \altc y} = \proj{n}{x} \altc \proj{n}{y}     & \axiom{PR5} \\
\proj{n}{\tau} = \tau                                   & \axiom{PR6} \\
\proj{n}{\tau \seqc x} = \tau \seqc \proj{n}{x}         & \axiom{PR7}
\end{axcol}
\end{eqntbl}
\end{table}
RDP and RSP are actually axiom schemas in which $X$ stands for an
arbitrary variable, $t_X$ stands for an arbitrary \ACPt\ term, and $E$
stands for an arbitrary guarded recursive specification over \ACPt.
Side conditions are added to restrict what $X$, $t_X$ and $E$ stand for.

Closed terms of \ACPt\ extended with constants for the components of
the solutions of guarded recursive specifications that denote the same
process cannot always be proved equal by means of the axioms of \ACPt\
together with RDP and RSP.
We introduce AIP to remedy this.
AIP is based on the view that two processes are identical if their
approximations up to any finite depth are identical.
The approximation up to depth $n$ of a process behaves the same as that
process, except that it cannot perform any further atomic action after
$n$ atomic actions have been performed.
AIP is given in Table~\ref{axioms-ACPt-recursion}.
Here, approximation up to depth $n$ is phrased in terms of a unary
\emph{projection} operator $\projop{n}$.
The axioms for the projection operators are axioms PR1--PR7 in
Table~\ref{axioms-ACPt-recursion}.
PR1--PR7 are actually axiom schemas in which $a$ stands for arbitrary
constants of \ACPt\ different from $\tau$ and $n$ stands for an
arbitrary natural number.

We write \ACPtrec\ for \ACPt\ extended with the constants for the
components of the solutions of guarded recursive specifications, the
projection operators, and the axioms RDP, RSP, AIP and PR1--PR7.

The minimal models of \ACPtrec\ are considered its standard models.
Henceforth, we restrict ourselves to the minimal models of \ACPtrec.
We assume that a fixed but arbitrary minimal model $\Macp$ of \ACPtrec\
has been given.

From Section~\ref{sect-expressiveness}, we will sometimes assume that
CFAR (Cluster Fair Abstraction Rule) is valid in $\Macp$.
CFAR says that a cluster of silent steps that has exits can be
eliminated if all exits are reachable from everywhere in the cluster.
A precise formulation of CFAR can be found in~\cite{Fok00}.

We use the term \emph{process} for the elements from the domain of
$\Macp$, and we denote the interpretations of constants and operators in
$\Macp$ by the constants and operators themselves.

Let $P$ be a process.
Then the set of \emph{states} or \emph{subprocesses} of $P$,
written $\Sub(P)$, is inductively defined as follows:
\begin{itemize}
\item
$P \in \Sub(P)$;
\item
if $e \seqc P' \in \Sub(P)$, then $P' \in \Sub(P)$;
\item
if $e \seqc P' \altc P'' \in \Sub(P)$, then $P' \in \Sub(P)$.
\end{itemize}

Let $P$ be a process and let $\Act' \subseteq \Actt$.
Then $P$ is \emph{regular over} $\Act'$ if the following conditions are
satisfied:
\begin{itemize}
\item
$\Sub(P)$ is finite;
\item
for all $P' \in \Sub(P)$ and $e \in \Actt$,
$e \seqc P' \in \Sub(P)$ implies $e \in \Act'$;
\item
for all $P',P'' \in \Sub(P)$ and $e \in \Actt$,
$e \seqc P' \altc P'' \in \Sub(P)$ implies $e \in \Act'$.
\end{itemize}
We say that $P$ is \emph{regular} if $P$ is regular over $\Actt$.

We will make use of the fact that being a regular process over $\Act$
coincides with being a component of the solution of a finite guarded
recursive specification in which the right-hand sides of the recursion
equations are linear terms.
\emph{Linearity} of terms is inductively defined as follows:
\begin{itemize}
\item
$\dead$ is linear;
\item
if $e \in \Actt$, then $e$ is linear;
\item
if $e \in \Actt$ and $X$ is a variable, then $e \seqc X$ is linear;
\item
if $t$ and $t'$ are linear, then $t \altc t'$  is linear.
\end{itemize}
A \emph{linear recursive specification} over \ACPt\ is a guarded
recursive specification $E = \set{X = t_X \where X \in V}$ over \ACPt,
where each $t_X$ is linear.
\begin{proposition}
\label{prop-lin-rec-process}
Let $P$ be a process and let $\Act' \subseteq \Act$.
Then $P$ is regular over $\Act'$ iff there exists a finite linear
recursive specification $E$ over \ACPt\ in which only atomic actions
from $\Act'$ occur such that $P$ is a component of the solution of $E$.
\end{proposition}
\begin{proof}
The proof follows the same line as the proof of
Proposition~\ref{prop-lin-rec-thread}.
\qed
\end{proof}
\begin{remark}
Proposition~\ref{prop-lin-rec-process} is concerned with processes that
are regular over $\Act$.
We can also prove that being a regular process over $\Actt$ coincides
with being a component of the solution of a finite linear recursive
specification over \ACPt if we assume that the cluster fair abstraction
rule~\cite{Fok00} holds in the model $\Macp$.
However, we do not need this more general result.
\end{remark}

We will write $\vAltc{i \in S} t_i$, where $S = \set{i_1,\ldots,i_n}$
and $t_{i_1},\ldots,t_{i_n}$ are \ACPt\ terms,
for $t_{i_1} \altc \ldots \altc t_{i_n}$.
The convention is that $\vAltc{i \in S} t_i$ stands for $\dead$ if
$S = \emptyset$.
We will often write $X$ for $\rec{X}{E}$ if $E$ is clear from the
context.
It should be borne in mind that, in such cases, we use $X$ as a
constant.

\section{Program-Service Interaction Instructions}
\label{sect-psii}

Recall that, in \PGA, it is assumed that a fixed but arbitrary set
$\BInstr$ of basic instructions has been given.
In the sequel, we will make use a version of \PGA\ in which the
following additional assumptions relating to $\BInstr$ are made:
\begin{itemize}
\item
a fixed but arbitrary finite set $\Foci$ of \emph{foci} has been given;
\item
a fixed but arbitrary finite set $\Meth$ of \emph{methods} has been
given;
\item
$\BInstr = \set{f.m \where f \in \Foci, m \in \Meth}$.
\end{itemize}

Each focus plays the role of a name of some service provided by an
execution environment that can be requested to process a command.
Each method plays the role of a command proper.
Executing a basic instruction of the form $f.m$ is taken as making a
request to the service named $f$ to process command $m$.

A basic instruction of the form $f.m$ is called a
\emph{program-service interaction instruction}.
Recall that, in \BTA, it is assumed that a fixed but arbitrary set
$\BAct$ of basic actions has been given.
In the sequel, we will make use of a version of \BTA\ in which
$\BAct = \BInstr$.
A basic action of the form $f.m$ is called a
\emph{thread-service interaction action}.

The intuition concerning program-service interaction instructions given
above will be made fully precise in Section~\ref{sect-process-extr},
using \ACP.

\section{Process Extraction}
\label{sect-process-extr}

In this section, we use \ACPtrec\ to make mathematically precise which
processes are produced by threads.

For that purpose, $\Act$ and $\commm$ are taken such that the following
conditions are satisfied:%
\footnote
{As usual, we will write $\Bool$ for the set $\set{\True,\False}$.}
\begin{ldispl}
\begin{aeqns}
\Act & \supseteq &
\set{\snd_f(d) \where f \in \Foci, d \in \Meth \union \Bool} \union
\set{\rcv_f(d) \where f \in \Foci, d \in \Meth \union \Bool}
\union
\set{\stp,\iact}
\end{aeqns}
\end{ldispl}%
and for all $f \in \Foci$, $d \in \Meth \union \Bool$, and
$e \in \Act$:
\begin{ldispl}
\begin{aeqns}
\snd_f(d) \commm \rcv_f(d) = \iact\;,
\\
\snd_f(d) \commm e = \dead & & \mif e \neq \rcv_f(d)\;,
\\
e \commm \rcv_f(d) = \dead & & \mif e \neq \snd_f(d)\;,
\end{aeqns}
\qquad\qquad
\begin{aeqns}
{} \phantom{\snd_f(d) \commm \rcv_f(d) = \iact\;,} \\
\stp \commm e = \dead & & \mif e \neq \stp\;,
\\
\iact \commm e = \dead\;.
\end{aeqns}
\end{ldispl}%
Actions of the forms $\snd_f(d)$ and $\rcv_f(d)$ are send and receive
actions, respectively, $\stp$ is an explicit termination action, and
$\iact$ is a concrete internal action.

The \emph{process extraction} operation $\pextr{\ph}$ assigns a process
to each thread.
The process extraction operation $\pextr{\ph}$ is defined by
$\pextr{T} = \abstr{\set{\stp}}(\cpextr{T})$, where $\cpextr{\ph}$ is
defined by the equations given in Table~\ref{eqns-process-extr}
(for $f \in \Foci$ and $m \in \Meth$).%
\begin{table}[!t]
\caption{Defining equations for process extraction operation}
\label{eqns-process-extr}
\begin{eqntbl}
\begin{eqncol}
\cpextr{\Stop} = \stp
\\
\cpextr{\DeadEnd} = \dead
\\
\cpextr{\pcc{T}{\Tau}{T'}} = \iact \seqc \iact \seqc \cpextr{T}
\\
\cpextr{\pcc{T}{f.m}{T'}} =
\snd_f(m) \seqc
(\rcv_f(\True) \seqc \cpextr{T} \altc
 \rcv_f(\False) \seqc \cpextr{T'})
\end{eqncol}
\end{eqntbl}
\end{table}

Let $P$ be a process, $T$ be a thread, and $F$ be an instruction
sequence.
Then we say that $T$ \emph{produces} $P$ if
$\tau \seqc \abstr{I}(\textr{T}) = \tau \seqc P$
for some $I \subseteq \Act$,
and we say that $F$ \emph{produces} $P$ if $\textr{F}$ produces $P$.

Notice that two atomic actions are involved in performing a basic action
of the form $f.m$: one for sending a request to process command $m$ to
the service named $f$ and another for receiving a reply from that
service upon completion of the processing.
Notice also that, for each thread $T$, $\cpextr{T}$ is a process that in
the event of termination performs a special termination action just
before termination.
Abstraction from this termination action yields the process denoted by
$\pextr{T}$.

The process extraction operation preserves the axioms of \BTArec.
Before we make this fully precise, we have a closer look at the axioms
of \BTArec.

A proper axiom is an equation or a conditional equation.
In Table~\ref{axioms-BTA-recursion}, we do not find proper axioms.
Instead of proper axioms, we find axiom schemas without side conditions
and axiom schemas with side conditions.
The axioms of \BTArec\ are obtained by replacing each axiom schema by
all its instances.

Henceforth, we write $\alpha^*$, where $\alpha$ is a valuation of
variables in $\Mbta$, for the unique homomorphic extension of $\alpha$
to terms of \BTArec.
Moreover, we identify $t_1 = t_2$ and $\emptyset \Limpl t_1 = t_2$.
\begin{proposition}
\label{prop-preservation-axioms}
Let $E \Limpl t_1 = t_2$ be an axiom of \BTArec,
and let $\alpha$ be a valuation of variables in $\Mbta$.
Then $\pextr{\alpha^*(t_1)} = \pextr{\alpha^*(t_2)}$ if
$\pextr{\alpha^*(t'_1)} = \pextr{\alpha^*(t'_2)}$ for all
$t'_1 = t'_2 \,\in\, E$.
\end{proposition}
\begin{proof}
The proof is trivial for the axiom of \BTA\ and the axioms RDP and RSP.
Using the equation $\cpextr{\proj{n}{T}} = \proj{2n}{\cpextr{T}}$, the
proof is also trivial for the axioms AIP and P0--P3.
This equation is easily proved by induction on $n$ and case distinction
on the structure of $T$ in both the basis step and the inductive step.
\qed
\end{proof}
\begin{remark}
Proposition~\ref{prop-preservation-axioms} would go through if no
abstraction of the above-men\-tioned special termination action was
made.
Notice further that \ACPt\ without the silent step constant and the
abstraction operator, better known as \ACP, would suffice if no
abstraction of the special termination action was made.
\end{remark}

\section{A Simple Protocol for Remote Instruction Processing}
\label{sect-simple-protocol}

In this section and the next section, we consider two protocols for
remote instruction processing.
The simple protocol described in this section is presumably the most
straightforward protocol for remote instruction processing that can be
achieved.
Therefore, we consider it a suitable starting-point for the design of
more advanced protocols for remote instruction processing -- such as
the one described in the next section.
Before this simple protocol is described, an extension of \ACP\ is
introduced to simplify the description of the protocols.

The following extension of \ACP\ from~\cite{BB92c} will be used: the
non-branching conditional operator ${}\gc{}$ over $\Bool$.
The expression $b \gc p$, is to be read as
\texttt{if} $b$ \texttt{then} $p$ \texttt{else} $\dead$.
The additional axioms for the non-branching conditional operator are
\begin{ldispl}
\True \gc x = x \quad \mathrm{and} \quad \False \gc x = \dead\;.
\end{ldispl}%
In the sequel, we will use expressions whose evaluation yields Boolean
values instead of the constants $\True$ and $\False$.
Because the evaluation of the expressions concerned are not dependent
on the processes denoted by the terms in which they occur, we will
identify each such expression with the constant for the Boolean value
that its evaluation yields.
Further justification of this can be found in~\cite[Section~9]{BM05a}.

The protocols concern systems whose main components are an
\emph{instruction stream generator} and an \emph{instruction stream
execution unit}.
The instruction stream generator generates different instruction
streams for different threads.
This is accomplished by starting it in different states.
The general idea of the protocols is that:
\begin{itemize}
\item
the instruction stream generator generating an instruction stream for a
thread $\pcc{T}{a}{T'}$ sends $a$ to the instruction stream execution
unit;
\item
on receipt of $a$, the instruction stream execution unit gets the
execution of $a$ done and sends the reply produced to the instruction
stream generator;
\item
on receipt of the reply, the instruction stream generator proceeds with
generating an instruction stream for $T$ if the reply is $\True$ and for
$T'$ otherwise.
\end{itemize}
In the case where the thread is $\Stop$ or $\DeadEnd$, the instruction
stream generator sends a special instruction ($\stopd$ or $\deadd$) and
the instruction stream execution unit does not send back a reply.

In this section, we consider a very simple protocol for remote
instruction processing that makes no effort to keep the execution unit
busy without intermission.

In the protocols, the generation of an instruction stream start from
the thread produced by an instruction sequence under execution instead
of the instruction sequence itself.
It follows immediately from the definition of the thread extraction
operation that the threads produced by instruction sequences under
execution are regular threads.
Therefore, we restrict ourselves to regular threads.

We write $\BActi$ for the set $\BAct \union \set{\stopd,\deadd}$.
Elements from $\BActi$ will loosely be called instructions.
The restriction of the domain of $\Mbta$ to the regular threads will
be denoted by $\Reg$.

The functions $\nm{act}$, $\nm{thrt}$, and $\nm{thrf}$ defined below
give, for each thread $T$ different from $\Stop$ and $\DeadEnd$, the
basic action that $T$ will perform first, the thread with which it will
proceed if the reply from the execution environment is $\True$, and the
thread with which it will proceed if the reply from the execution
environment is $\False$, respectively.
The functions $\funct{\nm{act}}{\Reg}{\BActi}$,
$\funct{\nm{thrt}}{\Reg}{\Reg}$, and $\funct{\nm{thrf}}{\Reg}{\Reg}$ are
defined as follows:
\begin{ldispl}
\begin{geqns}
\act{\Stop} = \stopd\;,
\\
\act{\DeadEnd} = \deadd\;,
\\
\act{\pcc{T}{a}{T'}} = a\;,
\end{geqns}
\quad
\begin{geqns}
\rest{\Stop} = \DeadEnd\;,
\\
\rest{\DeadEnd} = \DeadEnd\;,
\\
\rest{\pcc{T}{a}{T'}} = T\;,
\end{geqns}
\quad
\begin{geqns}
\resf{\Stop} = \DeadEnd\;,
\\
\resf{\DeadEnd} = \DeadEnd\;,
\\
\resf{\pcc{T}{a}{T'}} = T'\;.
\end{geqns}
\end{ldispl}%

The function $\enablea$ defined below is used by the instruction stream
generator to distinguish when it starts with handling the
instruction to be executed next between the different instructions that
it may be.
The function
$\funct{\enablea}
  {\BActi \x \Reg}{\Bool}$
is defined as follows:
\begin{ldispl}
\enablea(a,T) =
\left\{
\begin{array}{@{}l@{\quad}l@{}}
\True  & \mif \act{T} = a
\\
\False & \mif \act{T} \neq a\;.
\end{array}
\right.
\end{ldispl}%

For the purpose of describing the simple protocol outlined above in
\ACPt, $\Act$ and $\commm$ are taken such that, in addition to the
conditions mentioned at the beginning of
Section~\ref{sect-process-extr}, the following conditions are satisfied:
\begin{ldispl}
\begin{aeqns}
\Act & \supseteq &
\set{\snd_i(d) \where i \in \set{1,2}, d \in \BActi}    \union
\set{\rcv_i(d) \where i \in \set{1,2}, d \in \BActi}
\\ & {} \union {} &
\set{\snd_i(r) \where i \in \set{3,4}, r \in \Bool} \union
\set{\rcv_i(r) \where i \in \set{3,4}, r \in \Bool} \union
\set{\jact}
\end{aeqns}
\end{ldispl}%
and for all $i \in \set{1,2}$, $j \in \set{3,4}$, $d \in \BActi$,
$r \in \Bool$, and $e \in \Act$:
\begin{ldispl}
\begin{aeqns}
\snd_i(d) \commm \rcv_i(d) = \jact \;,
\\
\snd_i(d) \commm e = \dead & & \mif e \neq \rcv_i(d)\;,
\\
e \commm \rcv_i(d) = \dead & & \mif e \neq \snd_i(d)\;,
\eqnsep
\jact \commm e = \dead\;.
\end{aeqns}
\qquad\;
\begin{aeqns}
\snd_j(r) \commm \rcv_j(r) = \jact \;,
\\
\snd_j(r) \commm e = \dead & & \mif e \neq \rcv_j(r)\;,
\\
e \commm \rcv_j(r) = \dead & & \mif e \neq \snd_j(r)\;,
\end{aeqns}
\end{ldispl}%
Notice that the set $\Bool$ is the set of replies.

Let $T \in \Reg$.
Then the process representing the simple protocol for remote
instruction processing with regard to thread $T$ is described by
\begin{ldispl}
\encap{H}(\ISGA{T} \parc \IMTCA \parc \RTCA \parc \ISEUA)\;,
\end{ldispl}%
where the process $\ISGA{T}$ is recursively specified by the following
equation:
\begin{ldispl}
\renewcommand{\arraystretch}{1.4}
\begin{aeqns}
\ISGA{T} & = &
\smash{\Altc{f.m \in \BAct}}
 \enablea(f.m,T) \gc {}
\\ & & \phantom{\smash{\Altc{f.m \in \BAct}}} \quad
 \snd_1(f.m) \seqc
 (\rcv_4(\True)  \seqc \ISGA{\rest{T}} \altc
  \rcv_4(\False) \seqc \ISGA{\resf{T}})
\\ & \altc &
 \enablea(\stopd,T) \gc \snd_1(\stopd) \altc
 \enablea(\deadd,T) \gc \snd_1(\deadd)\;,
\end{aeqns}
\end{ldispl}%
the process $\IMTCA$ is recursively specified by the following equation:
\begin{ldispl}
\begin{aeqns}
\IMTCA & = &
\Altc{d \in \BActi}
 \rcv_1(d) \seqc \snd_2(d) \seqc \IMTCA\;,
\end{aeqns}
\end{ldispl}%
the process $\RTCA$ is recursively specified by the following equation:
\begin{ldispl}
\begin{aeqns}
\RTCA & = &
\Altc{r \in \Bool}
 \rcv_3(r) \seqc \snd_4(r) \seqc \RTCA\;,
\end{aeqns}
\end{ldispl}%
the process $\ISEUA$ is recursively specified by the following
equation:
\begin{ldispl}
\begin{aeqns}
\ISEUA & = &
\Altc{f.m \in \BAct}
 \rcv_2(f.m) \seqc \snd_f(m) \seqc
 (\rcv_f(\True)  \seqc \snd_3(\True) \altc
  \rcv_f(\False) \seqc \snd_3(\False)) \seqc
 \ISEUA
\\ & \altc &
\rcv_2(\stopd) \altc \rcv_2(\deadd) \seqc \dead
\end{aeqns}
\end{ldispl}%
and
\begin{ldispl}
\begin{aeqns}
H & = &
\set{\snd_i(d) \where i \in \set{1,2}, d \in \BActi}    \union
\set{\rcv_i(d) \where i \in \set{1,2}, d \in \BActi}
\\ & {} \union {} &
\set{\snd_i(r) \where i \in \set{3,4}, r \in \Bool} \union
\set{\rcv_i(r) \where i \in \set{3,4}, r \in \Bool}\;.
\end{aeqns}
\end{ldispl}%
$\ISGA{T}$ is the instruction stream generator for thread $T$, $\IMTCA$
is the transmission channel for messages containing instructions,
$\RTCA$ is the transmission channel for replies, and $\ISEUA$ is the
instruction stream execution unit.

If we abstract from all communications via the transmission channels,
then the process denoted by
$\encap{H}(\ISGA{T} \parc \IMTCA \parc \RTCA \parc \ISEUA)$
and the process $\pextr{T}$ are equal modulo an initial silent step.
\begin{theorem}
\label{theorem-simple-protocol}
For each $T \in \Reg$,
$\tau \seqc
 \abstr{\set{\jact}}
  (\encap{H}(\ISGA{T} \parc \IMTCA \parc \RTCA \parc \ISEUA))$
denotes the process $\tau \seqc \pextr{T}$.
\end{theorem}
\begin{proof}
Let $T \in \Reg$.
Moreover, let $E$ be a finite linear recursive specification over
\ACPt\ with $X \in \vars(E)$ such that $\pextr{T}$ is the $X$-component
of the solution of $E$ in $\Macp$.
By Proposition~\ref{prop-lin-rec-process} and the definition of
the process extraction operation, it is sufficient to prove that
\begin{ldispl}
\tau \seqc
\abstr{\set{\jact}}
 (\encap{H}(\ISGA{T} \parc \IMTCA \parc \RTCA \parc \ISEUA)) =
\tau \seqc \rec{X}{E}\;.
\end{ldispl}%
By AIP, it is sufficient to prove that for all $n \geq 0$:
\begin{ldispl}
\proj{n}{\tau \seqc
         \abstr{\set{\jact}}
          (\encap{H}
            (\ISGA{T} \parc \IMTCA \parc \RTCA \parc \ISEUA))} =
\proj{n}{\tau \seqc \rec{X}{E}}\;.
\end{ldispl}%
This is easily proved by induction on $n$ and in the inductive step by
case distinction on the structure of $T$, using the axioms of \ACPt\ and
RDP and in addition the fact that $\pextr{T'} \in \Sub(\pextr{T})$ for
all $T' \in \Res(T)$ and the fact that there exists an bijection between
$\Sub(\pextr{T})$ and $\vars(E)$.
\qed
\end{proof}

\section{A More Complex Protocol}
\label{sect-complex-protocol}

In this section, we consider a more complex protocol for remote
instruction processing that makes an effort to keep the execution unit
busy without intermission.

The specifics of the more complex protocol considered here are that:
\begin{itemize}
\item
the instruction stream generator may run ahead of the instruction stream
execution unit by not waiting for the receipt of the replies resulting
from the execution of instructions that it has sent earlier;
\item
to ensure that the instruction stream execution unit can handle the
run-ahead, each instruction sent by the instruction stream generator is
accompanied with the sequence of replies after which the instruction
must be executed;
\item
to correct for replies that have not yet reached the instruction stream
generator, each instruction sent is also accompanied with the number of
replies received since the last sending of an instruction.
\end{itemize}
This protocol is reminiscent of an instruction pre-fetching mechanism
as found in pipelined processors (see e.g.~\cite{HP03a}), but its range
of application is not restricted to pipelined instruction processing.

We write $\Bool^{\leq n}$, where $n \in \Nat$, for the set
$\set{u \in \seqof{\Bool} \where \len(u) \leq n}$.%
\footnote
{As usual, we write $\seqof{D}$ for the set of all finite sequences
 with elements from set $D$ and $\len(\sigma)$ for the length of finite
 sequence $\sigma$.
 Moreover, we write $\emptyseq$ for the empty sequence, $d$ for the
 sequence having $d$ as sole element, $\sigma \sigma'$ for the
 concatenation of finite sequences $\sigma$ and $\sigma'$, and
 $\tail(\sigma)$ for the tail of finite sequence $\sigma$.
}

It is assumed that a natural number $\maxlen$ has been given.
The number $\maxlen$ is taken for the maximal number of steps that the
instruction stream generator may run ahead of the instruction stream
execution unit.
Whether the execution unit can be kept busy without intermission with
the given $\maxlen$ depends on the actual execution times of
instructions and the actual transmission times over the transmission
channels involved.
If the execution unit can be kept busy without intermission with the
given $\maxlen$, then it is useless to increase~$\maxlen$.

The set $\Msg$ of \emph{instruction messages} is defined as follows:
\begin{ldispl}
\Msg = [0,\maxlen] \x \Bool^{\leq \maxlen} \x \BActi\;.
\end{ldispl}%
In an instruction message $\tup{n,u,a} \in \Msg$:
\begin{itemize}
\item
$n$ is the number of replies that are acknowledged by the message;
\item
$u$ is the sequence of replies after which the instruction that is part
of the message must be executed;
\item
$a$ is the instruction that is part of the message.
\end{itemize}
The instruction stream generator sends instruction messages via an
instruction message transmission channel to the instruction stream
execution unit.
We refer to a succession of transmitted instruction messages as an
\emph{instruction stream}.
An instruction stream is dynamic by nature, in contradistinction with an
instruction sequence.

The set $\Stisg$ of \emph{instruction stream generator states} is
defined as follows:
\begin{ldispl}
\Stisg = [0,\maxlen] \x \setof{(\Bool^{\leq \maxlen+1} \x \Reg)}\;.
\end{ldispl}%
In an instruction stream generator state $\tup{n,R} \in \Stisg$:
\begin{itemize}
\item
$n$ is the number of replies that has been received by the instruction
stream generator since the last acknowledgement of received replies;
\item
in each $\tup{u,T} \in R$, $u$ is the sequence of replies after which
the thread $T$ must be performed.
\end{itemize}
The functions $\updpm$ and $\updcr$ defined below are used to model the
updates of the instruction stream generator state on producing a message
and consuming a reply, respectively.
The function
$\funct{\updpm}{(\Bool^{\leq \maxlen} \x \Reg) \x \Stisg}{\Stisg}$
is defined as follows:
\begin{ldispl}
\updpm(\tup{u,T},\tup{n,R}) =
\\ \quad \left\{
\begin{array}{@{}l@{\;}l@{}}
\tup{0,(R \diff \set{\tup{u,T}}) \union
       \set{\tup{u\True,\rest{T}},\tup{u\False,\resf{T}}}} &
\mif \act{T} \in \BAct
\\
\tup{0,(R \diff \set{\tup{u,T}})} &
\mif \act{T} \notin \BAct\;.
\end{array}
\right.
\end{ldispl}%
The function $\funct{\updcr}{\Bool \x \Stisg}{\Stisg}$ is defined as
follows:
\begin{ldispl}
\updcr(r,\tup{n,R}) =
\tup{n + 1,\set{\tup{u,T} \where \tup{ru,T} \in R}}\;.
\end{ldispl}%
The function $\select$ defined below is used to model the selection of
the sequence of replies and the instruction that will be part of the
next message produced by the instruction stream generator.
The function
$\funct{\select}{\setof{(\Bool^{\leq \maxlen} \x \Reg)}}
                {\setof{(\Bool^{\leq \maxlen} \x \Reg)}}$
is defined as follows:
\begin{ldispl}
\select(R) =
\set{\tup{u,T} \in R \where
     \Forall{\tup{v,T'} \in R}{\len(u) \leq \len(v)}}\;.
\end{ldispl}%
Notice that $\tup{u,T} \in \select(R)$ and $\tup{v,T'} \in R$ only if
$\len(u) \leq \len(v)$.
By that breadth-first run-ahead is enforced.
The performance of the protocol would change considerably if
breadth-first run-ahead was not enforced.

The set $\Stiseu$  of \emph{instruction stream execution unit states} is
defined as follows:
\begin{ldispl}
\Stiseu = [0,\maxlen] \x \setof{(\Bool^{\leq \maxlen} \x \BActi)}\;.
\end{ldispl}%
In an instruction stream execution unit state $\tup{n,S} \in \Stiseu$:
\begin{itemize}
\item
$n$ is the number of replies for which the instruction stream
execution unit still has to receive an acknowledgement;
\item
in each $\tup{u,a} \in S$, $u$ is the sequence of replies after which
the instruction $a$ must be executed.
\end{itemize}
The functions $\updcm$ and $\updpr$ defined below are used to model the
updates of the instruction stream execution unit state on consuming a
message and producing a reply, respectively.
The function $\funct{\updcm}{\Msg \x \Stiseu}{\Stiseu}$ is defined as
follows:
\begin{ldispl}
\updcm(\tup{k,u,a},\tup{n,S}) =
\tup{n \monus k,S \union \set{\tup{\tail^{n \monus k}(u),a}}}\;.%
\footnotemark
\end{ldispl}%
\footnotetext
{As usual, we write $i \monus j$ for the monus of $i$ and $j$, i.e.\
 $i \monus j = i - j$ if $i \geq j$ and $i \monus j = 0$ otherwise.
 As usual, $\tail^n(u)$ is defined by induction on $n$ as follows:
 $\tail^0(u) = u$ and $\tail^{n+1}(u) = \tail(\tail^{n}(u))$.
}%
The function $\funct{\updpr}{\Bool \x \Stiseu}{\Stiseu}$ is defined as
follows:
\begin{ldispl}
\updpr(r,\tup{n,S}) =
\tup{n + 1,\set{\tup{u,a} \where \tup{ru,a} \in S}}\;.
\end{ldispl}%
The function $\enable$ defined below is used by the instruction stream
execution unit to distinguish when it starts with handling the
instruction to be executed next between the different instructions that
it may be.
The function
$\funct{\enable}
  {\BActi \x \setof{(\Bool^{\leq \maxlen} \x \BActi)}}{\Bool}$
is defined as follows:
\begin{ldispl}
\enable(a,S) =
\left\{
\begin{array}{@{}l@{\quad}l@{}}
\True  & \mif \tup{\emptyseq,a} \in S
\\
\False & \mif \tup{\emptyseq,a} \notin S\;.
\end{array}
\right.
\end{ldispl}%

The instruction stream execution unit sends replies via a reply
transmission channel to the instruction stream generator.
We refer to a succession of transmitted replies as a
\emph{reply stream}.

For the purpose of describing the transmission protocol in \ACPt, $\Act$
and $\commm$ are taken such that, in addition to the conditions
mentioned at the beginning of Section~\ref{sect-process-extr}, the
following conditions are satisfied:
\begin{ldispl}
\begin{aeqns}
\Act & \supseteq &
\set{\snd_i(d) \where i \in \set{1,2}, d \in \Msg} \union
\set{\rcv_i(d) \where i \in \set{1,2}, d \in \Msg}
\\ & {} \union {} &
\set{\snd_i(r) \where i \in \set{3,4}, r \in \Bool} \union
\set{\rcv_i(r) \where i \in \set{3,4}, r \in \Bool} \union
\set{\jact}
\end{aeqns}
\end{ldispl}%
and for all $i \in \set{1,2}$, $j \in \set{3,4}$, $d \in \Msg$,
$r \in \Bool$, and $e \in \Act$:
\begin{ldispl}
\begin{aeqns}
\snd_i(d) \commm \rcv_i(d) = \jact \;,
\\
\snd_i(d) \commm e = \dead & & \mif e \neq \rcv_i(d)\;,
\\
e \commm \rcv_i(d) = \dead & & \mif e \neq \snd_i(d)\;,
\eqnsep
\jact \commm e = \dead\;.
\end{aeqns}
\qquad\;
\begin{aeqns}
\snd_j(r) \commm \rcv_j(r) = \jact \;,
\\
\snd_j(r) \commm e = \dead & & \mif e \neq \rcv_j(r)\;,
\\
e \commm \rcv_j(r) = \dead & & \mif e \neq \snd_j(r)\;,
\end{aeqns}
\end{ldispl}%

Let $T \in \Reg$.
Then the process representing the more complex protocol for remote
instruction processing with regard to thread $T$ is described by
\begin{ldispl}
\encap{H}(\ISG{T} \parc \IMTC \parc \RTC \parc \ISEU)\;,
\end{ldispl}%
where the process $\ISG{T}$ is recursively specified by the following
equations:
\begin{ldispl}
\begin{aeqns}
\ISG{T} & = & \ISGi{\tup{0,\set{\tup{\emptyseq,T}}}}\;,
\eqnsep
\ISGi{\tup{n,R}} & = &
\Altc{\tup{u,T} \in \select(R)}
 \snd_1(\tup{n,u,\act{T}}) \seqc \ISGi{\updpm(\tup{u,T},\tup{n,R})}
\\ & \altc &
\xAltc{r \in \Bool}{\tup{u,T} \in \select(R)}
 \rcv_4(r)  \seqc \ISGi{\updcr(r,\tup{n,R})}
\\
\multicolumn{3}{@{}l@{}}
{(\mathrm{for\; every\;} \tup{n,R} \in \Stisg
  \mathrm{\;with\;} R \neq \emptyset)\;,}
\eqnsep
\ISGi{\tup{n,\emptyset}} & = & \jact
\\
\multicolumn{3}{@{}l@{}}
{(\mathrm{for\; every\;} \tup{n,\emptyset} \in \Stisg)\;,}
\end{aeqns}
\end{ldispl}%
the process $\IMTC$ is recursively specified by the following equation:
\begin{ldispl}
\begin{aeqns}
\IMTC & = &
\Altc{d \in \Msg}
 \rcv_1(d) \seqc \snd_2(d) \seqc \IMTC\;,
\end{aeqns}
\end{ldispl}%
the process $\RTC$ is recursively specified by the following equation:
\begin{ldispl}
\begin{aeqns}
\RTC & = &
\xAltc{r \in \Bool}{d \in \Msg}
 \rcv_3(r) \seqc \snd_4(r) \seqc \RTC\;,
\end{aeqns}
\end{ldispl}%
the process $\ISEU$ is recursively specified by the following
equations:
\begin{ldispl}
\begin{aeqns}
\ISEU & = & \ISEUi{\tup{0,\emptyset}}\;,
\eqnsep
\ISEUi{\tup{n,S}} & = &
\xAltc{d \in \Msg}{f.m \in \BAct}
 \rcv_2(d) \seqc \ISEUi{\updcm(d,\tup{n,S})}
\\ & \altc &
\Altc{f.m \in \BAct}
 \enable(f.m,S) \gc \snd_f(m) \seqc \ISEUii{\tup{f,\tup{n,S}}}
\\ & \altc &
\enable(\stopd,S) \gc \jact \altc
\enable(\deadd,S) \gc \dead
\\
\multicolumn{3}{@{}l@{}}
{(\mathrm{for\; every\;} \tup{n,S} \in \Stiseu)\;,}
\eqnsep
\ISEUii{\tup{f,\tup{n,S}}} & = &
\xAltc{r \in \Bool}{d \in \Msg}
 \rcv_f(r) \seqc \snd_3(r) \seqc \ISEUi{\updpr(r,\tup{n,S})}
\\ & \altc &
\Altc{d \in \Msg}
 \rcv_2(d) \seqc \ISEUii{\tup{f,\updcm(d,\tup{n,S})}}
\\
\multicolumn{3}{@{}l@{}}
{(\mathrm{for\; every\;} \tup{f,\tup{n,S}} \in \Foci \x \Stiseu)\;,}
\end{aeqns}
\end{ldispl}%
and
\begin{ldispl}
\begin{aeqns}
H & = &
\set{\snd_i(d) \where i \in \set{1,2}, d \in \Msg} \union
\set{\rcv_i(d) \where i \in \set{1,2}, d \in \Msg}
\\ & {} \union {} &
\set{\snd_i(r) \where i \in \set{3,4}, r \in \Bool} \union
\set{\rcv_i(r) \where i \in \set{3,4}, r \in \Bool}\;.
\end{aeqns}
\end{ldispl}%
$\ISG{T}$ is the instruction stream generator for thread $T$, $\IMTC$ is
the transmission channel for instruction messages, $\RTC$ is the
transmission channel for replies, and $\ISEU$ is the instruction stream
execution unit.

The protocol described above has been designed such that,
for each $T \in \Reg$,
$\tau \seqc
 \abstr{\set{\jact}}
  (\encap{H}(\ISG{T} \parc \IMTC \parc \RTC \parc \ISEU))$
denotes the process $\tau \seqc \pextr{T}$.
We refrain from presenting a proof of the claim that the protocol
satisfies this because this paper is first and foremost a conceptual
paper and the proof is straightforward but tedious.

The transmission channels $\IMTC$ and $\RTC$ can keep one instruction
message and one reply, respectively.
The protocol has been designed in such a way that the protocol will also
work properly if these channels are replaced by channels with larger
capacity and even by channels with unbounded capacity.

Suppose that the transmission times over the transmission channels are
small compared with the execution times of instructions.
Even then the protocol described in Section~\ref{sect-simple-protocol}
will always have to idle for a short time after the execution of an
instruction, whereas after an initial phase the protocol described above
will never have to idle after the execution of an instruction if the
instruction stream generator may run a few steps ahead of the
instruction stream execution unit.

\section{Adaptations of the Protocol}
\label{sect-adaptations}

In this section, we discuss some conceivable adaptations of the protocol
described in Section~\ref{sect-complex-protocol}.
While we were thinking through the details of that protocol, various
variations suggested themselves.
The variations discussed below are among the most salient ones.
We think they deserve mention.
However, their discussion is not in depth.
The reason for this is that these variations have not yet been
investigated thoroughly.

Consider the case where, for each instruction, it is known what the
probability is with which its execution leads to the reply $\True$.
This might give reason to adapt the protocol described in
Section~\ref{sect-complex-protocol}.
Suppose that the instruction stream generator states do not only keep
the sequences of replies after which threads must be performed, but also
the sequences of instructions involved in producing those sequences of
replies.
Then the probability with which the sequences of replies will happen can
be calculated and several conceivable adaptations of the protocol to
this probabilistic knowledge are possible by mere changes in the
selection of the sequence of replies and the instruction that will be
part of the next instruction message produced by the instruction stream
generator.
Among those adaptations are:
\begin{itemize}
\item
restricting the instruction messages that are produced ahead to the ones
where the sequence of replies after which the instruction must be
executed will happen with a probability $\geq 0.50$, but sticking to
breadth-first run-ahead;
\item
restricting the instruction messages that are produced ahead to the ones
where the sequence of replies after which the instruction must be
executed will happen with a probability $\geq 0.95$, but not sticking to
breadth-first run-ahead.
\end{itemize}
At first sight, these adaptations are reminiscent of combinations of an
instruction pre-fetching mechanism and a branch prediction mechanism as
found in pipelined processors (see e.g.~\cite{HP03a}).
However, usually branch prediction mechanisms make use of statistics
based on recently processed instructions instead of probabilistic
knowledge of the kind used in the protocols sketched above.

Regular threads can be represented in such a way that it is effectively
decidable whether the two threads with which a thread may proceed after
performing its first action are identical.
Consider the case where threads are represented in the instruction
stream generator states in such a way.
Then the protocol can be adapted such that no duplication of instruction
messages takes place in the cases where the two threads with which a
thread possibly proceeds after performing its first action are
identical.
This can be accomplished by using sequences of elements from
$\Bool \union \set{*}$, instead of sequences of elements from $\Bool$,
in instruction messages, instruction stream generator states, and
instruction stream execution unit states.
The occurrence of $*$ at position $i$ in a sequence indicates that the
$i$th reply may be either $\True$ or $\False$.
The impact of this change on the updates of instruction stream generator
states and instruction stream execution unit states is minor.
This adaptation is reminiscent of an instruction pre-fetch mechanism as
found in pipelined processors that prevents instruction pre-fetches
that are superfluous due to identity of branches.

\section{Alternative Choice Instructions}
\label{sect-pci}

Process algebra is an area of the study of concurrency which is
considered relevant to computer science, as is witnesses by the extent
of the work on algebraic theories of processes such as \ACP, CCS and
CSP in theoretical computer science.
This strongly hints that there are programmed systems whose behaviours
can be taken for processes as considered in process algebra.
Therefore, it is interesting to know to which extent the behaviours
considered in process algebra can be produced by programs under
execution, starting from the perception of a program as an instruction
sequence.
In coming sections, we will establish results concerning the processes
as considered in \ACP\ that can be produced by instruction sequences
under execution.

For the purpose of producing processes as considered in \ACP, we need a
version of \PGA\ with special basic instructions to deal with the
non-deterministic choice between alternatives that stems from the
alternative composition of processes.
Recall that, in \PGA, it is assumed that a fixed but arbitrary set
$\BInstr$ of basic instructions has been given.
In the coming sections, we will make use a version of \PGA\ in which the
following additional assumptions relating to $\BInstr$ are made:
\begin{itemize}
\item
a fixed but arbitrary finite set $\Foci$ of \emph{foci} has been given;
\item
a fixed but arbitrary finite set $\Meth$ of \emph{methods} has been
given;
\item
a fixed but arbitrary set $\AAct$ of \emph{atomic actions}, with
$\tact \notin \AAct$, has been given;
\item
$\BInstr =
 \set{f.m \where f \in \Foci, m \in \Meth} \union
 \set{\ac(e_1,e_2) \where e_1,e_2 \in \AAct \union \set{\tact}}$.
\end{itemize}

On execution of a basic instruction $\ac(e_1,e_2)$, first a
non-deterministic choice between the atomic actions $e_1$ and $e_2$ is
made and then the chosen atomic action is performed.
The reply $\True$ is produced if $e_1$ is performed and the reply
$\False$ is produced if $e_2$ is performed.
Basic instructions of this kind are material to produce all regular
processes by means of instruction sequences.
A basic instruction of the form $\ac(e_1,e_2)$ is called an
\emph{alternative choice instruction}.
Henceforth, we will write \PGAac\ for the version of \PGA\ with
alternative choice instructions.

The intuition concerning alternative choice instructions given above
will be made fully precise at the end of this section, using \ACPt.
It will not be made fully precise using an extension of \BTA\ because it
is considered a basic property of threads that they are deterministic
behaviours.

Recall that we  make use of a version of \BTA\ in which
$\BAct = \BInstr$.
A basic action of the form $\ac(e_1,e_2)$ is called an
\emph{alternative choice action}.
Henceforth, we will write \BTAac\ for the version of \BTA\ with
alternative choice actions.

For the purpose of making precise what processes are produced by the
threads denoted by closed terms of \BTAacrec, $\Act$ and $\commm$ are
taken such that, in addition to the conditions mentioned at the
beginning of Section~\ref{sect-process-extr}, the following conditions
are satisfied:
\begin{ldispl}
\begin{aeqns}
\Act & \supseteq & \AAct \union \set{\tact}
\end{aeqns}
\end{ldispl}%
and for all $e,e' \in \Act$:
\begin{ldispl}
\begin{aeqns}
e' \commm e = \dead & & \mif e' \in \AAct \union \set{\tact}\;.
\end{aeqns}
\end{ldispl}%

The process extraction operation for \BTAac\ has as defining equations
the equations given in Table~\ref{eqns-process-extr} and
in addition the equation given in Table~\ref{eqns-process-extr-add-pc}.
\begin{table}[!t]
\caption{Additional defining equation for process extraction operation}
\label{eqns-process-extr-add-pc}
\begin{eqntbl}
\begin{eqncol}
\cpextr{\pcc{T}{\ac(e,e')}{T'}} =
e \seqc \cpextr{T} \altc e' \seqc \cpextr{T'}
\end{eqncol}
\end{eqntbl}
\end{table}

Proposition~\ref{prop-preservation-axioms} goes through for \BTAac.

\section{Instruction Sequence Producible Processes}
\label{sect-expressiveness}

It follows immediately from the definitions of the thread extraction and
process extraction operations that the instruction sequences considered
in \PGA\ produce regular processes.
The question is whether all regular processes are producible by these
instruction sequences.
In this section, we show that all regular processes can be produced by
the instruction sequences with alternative choice instructions.

We will make use of the fact that all regular threads over $\BAct$ can
be produced by the single-pass instruction sequences considered in \PGA.
\begin{proposition}
\label{prop-expressiveness}
For each thread $T$ that is regular over $\BAct$, there exists a \PGA\
instruction sequence $F$ such that $F$ produces $T$, i.e.\
$\textr{F} = T$.
\end{proposition}
\begin{proof}
By Proposition~\ref{prop-lin-rec-thread}, $T$ is a component of the
solution of some finite linear recursive specification $E$ over \BTA.
There occur finitely many variables $X_0,\ldots,X_n$ in $E$.
Assume that $T$ is the $X_0$-component of the solution of $E$.
Let $F$ be the \PGA\ instruction sequence
$(F_0 \conc \ldots \conc F_n)\rep$, where $F_i$ is defined as follows
($0 \leq i \leq n$):

\begin{scriptsize}
\begin{ldispl}
\mbox{} \hsp{-2.5}
F_i =
\left\{
\begin{array}[c]{@{}l@{\;}l@{}}
\halt \conc \halt \conc \halt
 & \mif X_i = \Stop \in E \\
\fjmp{0} \conc \fjmp{0} \conc \fjmp{0}
 & \mif X_i = \DeadEnd \in E \\
\ptst{a} \conc \fjmp{3{\cdot}(j{-}i){-}1} \conc
               \fjmp{3{\cdot}(k{-}i){-}2}
 & \mif X_i = \pcc{X_j}{a}{X_k} \in E \Land i < j \Land i < k \\
\ptst{a} \conc \fjmp{3{\cdot}(j{-}i){-}1} \conc
               \fjmp{3{\cdot}(n{+}1{-}(i{-}k)){-}2}
 & \mif X_i = \pcc{X_j}{a}{X_k} \in E \Land i < j \Land i \geq k \\
\ptst{a} \conc \fjmp{3{\cdot}(n{+}1{-}(i{-}j)){-}1} \conc
               \fjmp{3{\cdot}(k{-}i){-}2}
 & \mif X_i = \pcc{X_j}{a}{X_k} \in E \Land i \geq j \Land i < k \\
\ptst{a} \conc \fjmp{3{\cdot}(n{+}1{-}(i{-}j)){-}1} \conc
               \fjmp{3{\cdot}(n{+}1{-}(i{-}k)){-}2}
 & \mif X_i = \pcc{X_j}{a}{X_k} \in E \Land i \geq j \Land i \geq k.
\end{array}
\right.
\end{ldispl}%
\end{scriptsize}%
\sloppy
Then $F$ is a \PGA\ instruction sequence such that the interpretation
of $\textr{F} = T$.
\qed
\end{proof}

All regular processes over $\AAct$ can be produced by the instruction
sequences considered in \PGAac.
\begin{theorem}
\label{theorem-completeness}
Assume that CFAR is valid in $\Macp$.
Then, for each process $P$ that is regular over $\AAct$, there exists an
instruction sequence $F$ in which only basic instructions of the form
$\ac(e,\tact)$ occur such that $F$ produces $P$, i.e.\
$\tau \seqc \abstr{\set{\tact}}(\pextr{\textr{F}}) = \tau \seqc P$.
\end{theorem}
\begin{proof}
By Propositions~\ref{prop-lin-rec-thread}, \ref{prop-lin-rec-process}
and~\ref{prop-expressiveness}, it is sufficient to show that, for each
finite linear recursive specification $E$ over \ACPt\ in which only
atomic actions from $\AAct$ occur, there exists a finite linear
recursive specification $E'$ over \BTAac\ in which only basic actions of
the form $\ac(e,\tact)$ occur such that
$\tau \seqc \rec{X}{E} =
 \tau \seqc \abstr{\set{\tact}}(\pextr{\rec{X}{E'}})$
for all $X \in \vars(E)$.

Take the finite linear recursive specification $E$ over \ACPt\ that
consists of the recursion equations
\begin{ldispl}
X_i = e_{i 1} \seqc X_{i 1} \altc \ldots \altc e_{i k_i} \seqc X_{i k_i}
       \altc
      e'_{i 1} \altc \ldots \altc e'_{i l_i}\;,
\end{ldispl}%
where
$e_{i 1},\ldots,e_{i k_i},e'_{i 1},\ldots,e'_{i l_i} \in
 \AAct$,
for $i \in \set{1,\ldots n}$.
Then construct the finite linear recursive specification $E'$ over
\BTAac\ that consists of the recursion equations
\begin{ldispl}
X_i =
\pcc{X_{i1}}{\ac(e_{i1},\tact)}
    {(\ldots(\pcc{X_{ik_i}}{\ac(e_{ik_i},\tact)}
                 {\\ \hsp{3.1} (\pcc{\Stop}{\ac(e'_{i1},\tact)}
                       {(\ldots(\pcc{\Stop}{\ac(e'_{il_i},\tact)}
                                    {X_i})\ldots)})})\ldots)}
\end{ldispl}%
for $i \in \set{1,\ldots n}$;
and the finite linear recursive specification $E''$ over \ACPt\ that
consists of the recursion equations
\begin{ldispl}
\begin{aeqns}
X_i       & = & e_{i 1}   \seqc X_{i 1}   \altc \tact \seqc Y_{i 2}\;,
\\
Y_{i 2}   & = & e_{i 2}   \seqc X_{i 2}   \altc \tact \seqc Y_{i 3}\;,
\\ & \vdots & \\
Y_{i k_i} & = & e_{i k_i} \seqc X_{i k_i} \altc \tact \seqc Z_{i 1}\;,
\end{aeqns}
\qquad \qquad
\begin{aeqns}
Z_{i 1}   & = & e'_{i 1}   \altc \tact \seqc Z_{i 2}\;,
\\
Z_{i 2}   & = & e'_{i 2}   \altc \tact \seqc Z_{i 3}\;,
\\ & \vdots & \\
Z_{i l_i} & = & e'_{i l_i} \altc \tact \seqc X_i\;,
\end{aeqns}
\end{ldispl}%
where $Y_{i 2},\ldots,Y_{i k_i},Z_{i 1},\ldots,Z_{i l_i}$ are fresh
variables, for $i \in \set{1,\ldots n}$.
It follows immediately from the definition of the process extraction
operation that $\pextr{\rec{X}{E'}} = \rec{X}{E''}$ for all
$X \in \vars(E)$.
Moreover, it follows from CFAR that
$\tau \seqc \rec{X}{E} = \tau \seqc \abstr{\set{\tact}}(\rec{X}{E''})$
for all $X \in \vars(E)$.
Hence,
$\tau \seqc \rec{X}{E} =
 \tau \seqc \abstr{\set{\tact}}(\pextr{\rec{X}{E'}})$
for all $X \in \vars(E)$.
\qed
\end{proof}
For example, assuming that CFAR is valid, the instruction sequence
\begin{ldispl}
(\ptst{\ac(\rcv_3(\True),\tact)}  \conc \fjmp{4} \conc
 \ptst{\ac(\rcv_3(\False),\tact)} \conc \fjmp{5} \conc \fjmp{7} \conc
\\ \phantom{(}
 \ptst{\ac(\snd_4(\True),\tact)}  \conc \fjmp{5} \conc \fjmp{9} \conc
 \ptst{\ac(\snd_4(\False),\tact)} \conc \fjmp{2} \conc \fjmp{9})\rep
\end{ldispl}%
produces the reply transmission channel process $\RTC$ of which a
guarded recursive specification is given in
Section~\ref{sect-complex-protocol}.

\begin{remark}
Theorem~\ref{theorem-completeness} with
``$\tau \seqc \abstr{\set{\tact}}(\pextr{\textr{F}}) = \tau \seqc P$''
replaced by ``$\pextr{\textr{F}} = P$'' can be established if \PGA\ is
extended with multiple-reply test instructions, see~\cite{BM08i}.
In that case, the assumption that CFAR is valid is superfluous.
\end{remark}

\section{Services and Use Operators}
\label{sect-services}

An instruction sequence under execution may make use of services.
That is, certain instructions may be executed for the purpose of having
the behaviour produced by the instruction sequence affected by a
service that takes those instructions as commands to be processed.
Likewise, a thread may perform certain actions for the purpose of having
itself affected by a service that takes those actions as commands to be
processed.
The processing of an action may involve a change of state of the service
and at completion of the processing of the action the service returns a
reply value to the thread.
The reply value determines how the thread proceeds.
The use operators can be used in combination with the thread extraction
operation from Section~\ref{sect-thread-extr} to describe the behaviour
produced by instruction sequences that make use of services.
In this section, we first review the use operators, which are concerned
with threads making such use of services, and then extend the process
extraction operation to the use operators.

A \emph{service} $H$ consists of
\begin{itemize}
\item
a set $S$ of \emph{states};
\item
an \emph{effect} function $\funct{\eff}{\Meth \x S}{S}$;
\item
a \emph{yield} function
$\funct{\yld}{\Meth \x S}{\Bool \union \set{\Blocked}}$;
\item
an \emph{initial state} $s_0 \in S$;
\end{itemize}
satisfying the following condition:
\begin{ldispl}
\Forall{m \in \Meth, s \in S}
{(\yld(m,s) = \Blocked \Limpl
  \Forall{m' \in \Meth}{\yld(m',\eff(m,s)) = \Blocked})}\;.
\end{ldispl}%
The set $S$ contains the states in which the service may be, and the
functions $\eff$ and $\yld$ give, for each method $m$ and state $s$, the
state and reply, respectively, that result from processing $m$ in state
$s$.
By the condition imposed on services, once the service has returned
$\Blocked$ as reply, it keeps returning $\Blocked$ as reply.

Let $H = \tup{S,\eff,\yld,s_0}$ be  a service and let $m \in \Meth$.
Then
the \emph{derived service} of $H$ after processing $m$, written
$\derive{m}H$, is the service $\tup{S,\eff,\yld,\eff(m,s_0)}$; and
the \emph{reply} of $H$ after processing $m$, written $H(m)$, is
$\yld(m,s_0)$.

When a thread makes a request to the service to process $m$:
\begin{itemize}
\item
if $H(m) \neq \Blocked$, then the request is accepted, the reply is
$H(m)$, and the service proceeds as $\derive{m}H$;
\item
if $H(m) = \Blocked$, then the request is rejected and the service
proceeds as a service that rejects any request.
\end{itemize}

We introduce the sort $\Serv$ of \emph{services}.
However, we will not introduce constants and operators to build terms
of this sort.
The sort $\Serv$, standing for the set of all services, is considered a
parameter of the extension of \BTA\ being presented.
Moreover, we introduce, for each $f \in \Foci$, the binary \emph{use}
operator $\funct{\useop{f}}{\Thr \x \Serv}{\Thr}$.
The axioms for these operators are given in Table~\ref{axioms-use}.%
\begin{table}[!t]
\caption{Axioms for use operators}
\label{axioms-use}
\begin{eqntbl}
\begin{saxcol}
\use{\Stop}{f}{H} = \Stop                           & & \axiom{U1} \\
\use{\DeadEnd}{f}{H} = \DeadEnd                     & & \axiom{U2} \\
\use{(\pcc{x}{\Tau}{y})}{f}{H} =
\pcc{(\use{x}{f}{H})}{\Tau}{(\use{y}{f}{H})}        & & \axiom{U3} \\
\use{(\pcc{x}{g.m}{y})}{f}{H} =
\pcc{(\use{x}{f}{H})}{g.m}{(\use{y}{f}{H})}
                               & \mif f \neq g        & \axiom{U4} \\
\use{(\pcc{x}{f.m}{y})}{f}{H} = \Tau \bapf (\use{x}{f}{\derive{m}H})
                               & \mif H(m) = \True    & \axiom{U5} \\
\use{(\pcc{x}{f.m}{y})}{f}{H} = \Tau \bapf (\use{y}{f}{\derive{m}H})
                               & \mif H(m) = \False   & \axiom{U6} \\
\use{(\pcc{x}{f.m}{y})}{f}{H} = \Tau \bapf \DeadEnd
                               & \mif H(m) = \Blocked & \axiom{U7} \\
\multicolumn{2}{@{}l@{\;\;}}{
\use{(\pcc{x}{\ac(e_1,e_2)}{y})}{f}{H} =
\pcc{(\use{x}{f}{H})}{\ac(e_1,e_2)}{(\use{y}{f}{H})}} & \axiom{U8} \\
\proj{n}{\use{x}{f}{H}} = \proj{n}{\use{\proj{n}{x}}{f}{H}}
                                                    & & \axiom{U9}
\end{saxcol}
\end{eqntbl}
\end{table}
Intuitively, $\use{T}{f}{H}$ is the thread that results from processing
all actions performed by thread $T$ that are of the form $f.m$ by
service $H$.
When a basic action of the form $f.m$ performed by thread $T$ is
processed by service $H$, it is turned into the basic action $\Tau$ and
postconditional composition is removed in favour of basic action
prefixing on the basis of the reply value produced.

We add the use operators to \PGAac\ as well.
We will only use the extension in combination with the thread extraction
operation $\textr{\ph}$ and define
$\textr{\use{F}{f}{H}} = \use{\textr{F}}{f}{H}$.
Hence, $\textr{\use{F}{f}{H}}$ denotes the thread produced by $F$ if $F$
makes use of $H$.
If $H$ is a service such as an unbounded counter, an unbounded stack or
a Turing tape, then a non-regular thread may be produced.

In order to extend the process extraction operation to the use
operators, we need an extension of \ACPt\ with action renaming operators
$\aren{h}$, where $\funct{h}{\Actt}{\Actt}$ such that $h(\tau) = \tau$.
The axioms for action renaming are given in~\cite{Fok00}.
Intuitively, $\aren{h}(P)$ behaves as $P$ with each atomic action
replaced according to $h$.
We write $\aren{e' \mapsto e''}$ for the renaming operator $\aren{h}$
with $h$ defined by $h(e') = e''$ and $h(e) = e$ if $e \neq e'$.

For the purpose of extending the process extraction operation to the use
operators, $\Act$ and $\commm$ are taken such that, in addition to the
conditions mentioned at the beginning of
Section~\ref{sect-process-extr}, with everywhere $\Bool$ replaced by
$\Bool \union \set{\Blocked}$, and the conditions mentioned at the end
of Section~\ref{sect-pci}, the following conditions are satisfied:
\begin{ldispl}
\begin{aeqns}
\Act & \supseteq &
\set{\snd_\serv(r) \where r \in \Bool \union \set{\Blocked}} \union
\set{\rcv_\serv(m) \where m \in \Meth} \union
\set{\stp^*}
\end{aeqns}
\end{ldispl}%
and for all $e \in \Act$, $m \in \Meth$, and
$r \in \Bool \union \set{\Blocked}$:
\begin{ldispl}
\begin{aeqns}
\snd_\serv(r) \commm e = \dead\;,
\\
e \commm \rcv_\serv(m) = \dead\;,
\end{aeqns}
\qquad\qquad
\begin{aeqns}
\stp \commm \stp = \stp^*\;,
\\
\stp^* \commm e = \dead\;.
\end{aeqns}
\end{ldispl}%

We also need to define a set $A_f \subseteq \Act$ and a function
$\funct{h_f}{\Actt}{\Actt}$ for each $f \in \Foci$:
\begin{ldispl}
A_f =
\set{\snd_f(d) \where d \in \Meth \union \Bool \union \set{\Blocked}}
 \union
\set{\rcv_f(d) \where d \in \Meth \union \Bool \union \set{\Blocked}}\;;
\end{ldispl}%
for all $e \in \Actt$, $m \in \Meth$ and
$r \in \Bool \union \set{\Blocked}$:
\begin{ldispl}
\begin{aeqns}
h_f(\snd_\serv(r)) = \snd_f(r)\;, \\
h_f(\rcv_\serv(m)) = \rcv_f(m)\;, \\
h_f(e)             = e
 & & \mif \LAND{r' \in \Nat}{e \neq \snd_\serv(r')} \land
          \LAND{m' \in \Meth}{e \neq \rcv_\serv(m')}\;.
\end{aeqns}
\end{ldispl}%

To extend the process extraction operation to the use operators, the
defining equation concerning the postconditional composition operators
has to be adapted and a new defining equation concerning the use
operators has to be added.
These two equations are given in Table~\ref{eqns-process-extr-add},%
\begin{table}[!t]
\caption{Adapted and additional defining equations for process
  extraction operation}
\label{eqns-process-extr-add}
\begin{eqntbl}
\begin{eqncol}
\cpextr{\pcc{T}{f.m}{T'}} =
\snd_f(m) \seqc
(\rcv_f(\True) \seqc \cpextr{T} \altc
 \rcv_f(\False) \seqc \cpextr{T'} \altc
 \rcv_f(\Blocked) \seqc \dead)
\eqnsep
\cpextr{\use{T}{f}{H}} =
\aren{\stp^* \mapsto \stp}
 (\encap{\set{\stp}}
   (\encap{{A_f}}(\cpextr{T} \parc \aren{{h_f}}(\cpextr{H}))))
\end{eqncol}
\end{eqntbl}
\end{table}
where $\cpextr{H}$ is the $X_H$-component of the solution of
\begin{ldispl}
\{X_{H'} =
  \Altc{m \in \Meth}
   \rcv_\serv(m) \seqc \snd_\serv(H'(m)) \seqc X_{\derive{m} H'} \altc
  \stp \where H' \in \rDelta(H)\}\;,
\end{ldispl}%
where $\rDelta(H)$ is inductively defined as follows:
\begin{itemize}
\item
$H \in \rDelta(H)$;
\item
if $m \in \Meth$ and $H' \in \rDelta(H)$, then
$\derive{m} H' \in \rDelta(H)$.
\end{itemize}

The extended process extraction operation preserves the axioms for the
use operators.
Owing to the presence of axiom schemas with semantic side conditions in
Table~\ref{axioms-use}, the axioms for the use operators include proper
axioms, which are all of the form $t_1 = t_2$, and axioms that have a
semantic side condition, which are all of the form
$t_1 = t_2\;\mif H(m) = r$.
By that, the precise formulation of the preservation result is somewhat
complicated.
\pagebreak[2]
\begin{proposition}
\label{prop-preservation-axioms-use}
\mbox{}
\begin{enumerate}
\item
Let $t_1 = t_2$ be a proper axiom for the use operators,
and let $\alpha$ be a valuation of variables in $\Mbta$.
Then $\pextr{\alpha^*(t_1)} = \pextr{\alpha^*(t_2)}$.
\item
Let $t_1 = t_2\;\mif H(m) = r$ be an axiom with semantic side condition
for the use operators, and let $\alpha$ be a valuation of variables in
$\Mbta$.
Then $\pextr{\alpha^*(t_1)} = \pextr{\alpha^*(t_2)}$ if $H(m) = r$.
\end{enumerate}
\end{proposition}
\begin{proof}
The proof is straightforward.
We sketch the proof for axiom U5.
By the definition of the process extraction operation, it is sufficient
to show that
$\cpextr{\use{(\pcc{T}{f.m}{T'})}{f}{H}} =
 \cpextr{\Tau \bapf (\use{T}{f}{\derive{m}H})}$
if $H(m) = \True$.
In outline, this goes as follows:
\begin{trivlist}
\item
$
\begin{geqns}
\cpextr{\use{(\pcc{T}{f.m}{T'})}{f}{H}}
\\ \, {} =
\aren{\stp^* \mapsto \stp}
\\ \, \phantom{{} = {}} \;\,
 (\encap{\set{\stp}}
   (\encap{A_f}
     (\snd_f(m) \seqc
      (\rcv_f(\True) \seqc \cpextr{T} \altc
       \rcv_f(\False) \seqc \cpextr{T'} \altc
       \rcv_f(\Blocked) \seqc \dead) \parc
      \aren{h_f}(\cpextr{H}))))
\\ \, {} =
\iact \seqc \iact \seqc
\aren{\stp^* \mapsto \stp}
 (\encap{\set{\stp}}
   (\encap{A_f}(\cpextr{T} \parc \aren{h_f}(\cpextr{\derive{m}H}))))
\\ \, {} =
\cpextr{\Tau \bapf (\use{T}{f}{\derive{m}H})}\;.
\end{geqns}
$
\end{trivlist}
In the first and third step, we apply defining equations of
$\cpextr{\ph}$.
In the second step, we apply axioms of \ACPtrec\ with action renaming,
and use that $H(m) = \True$.
\qed
\end{proof}

\begin{remark}
Let $F$ be a \PGAac\ instruction sequence and $H$ be a service.
Then $\pextr{\textr{\use{F}{f}{H}}}$ is the process produced by $F$ if
$F$ makes use of $H$.
Instruction sequences that make use of services such as unbounded
counters, unbounded stacks or Turing tapes are interesting because they
may produce non-regular processes.
\end{remark}

\section{PGLD Programs and the Use of Boolean Registers}
\label{sect-PGLD-BR}

In this section, we show that all regular processes can also be produced
by programs written in a program notation which is close to existing
assembly languages, and even by programs in which no atomic action
occurs more than once in an alternative choice instruction.
The latter result requires programs that make use of Boolean registers.

A hierarchy of program notations rooted in \PGA\ is introduced
in~\cite{BL02a}.
One program notation that belongs to this hierarchy is \PGLD, a very
simple program notation which is close to existing assembly languages.
It has absolute jump instructions and no explicit termination
instruction.

In \PGLD, like in \PGA, it is assumed that there is a fixed but
arbitrary finite set of \emph{basic instructions} $\BInstr$.
The primitive instructions of \PGLD\ differ from the primitive
instructions of \PGA\ as follows: for each $l \in \Nat$, there is
an \emph{absolute jump instruction}~$\ajmp{l}$ instead of a forward jump
instruction~$\fjmp{l}$.
\PGLD\ programs have the form $u_1;\ldots;u_k$, where $u_1,\ldots,u_k$
are primitive instructions of \PGLD.

The effects of all instructions in common with \PGA\ are as in \PGA\
with one difference: if there is no next instruction to be executed,
termination occurs.
The effect of an absolute jump instruction $\ajmp{l}$ is that execution
proceeds with the $l$-th instruction of the program concerned.
If $\ajmp{l}$ is itself the $l$-th instruction, then inaction occurs.
If $l$ equals $0$ or $l$ is greater than the length of the program, then
termination occurs.

We define the meaning of \PGLD\ programs by means of a function
$\pgldpga$ from the set of all \PGLD\ programs to the set of all
closed \PGA\ terms.
This function is defined by
\begin{ldispl}
\pgldpga(u_1 \conc \ldots \conc u_k) =
(\phi_1(u_1) \conc \ldots \conc \phi_k(u_k) \conc
 \halt \conc \halt)\rep\;,
\end{ldispl}%
where the auxiliary functions $\phi_j$ from the set of all primitive
instructions of \PGLD\ to the set of all primitive instructions of
\PGA\ are defined as follows ($1 \leq j \leq k$):
\begin{ldispl}
\begin{aceqns}
\phi_j(\ajmp{l}) & = & \fjmp{l-j}       & \mif j \leq l \leq k\;, \\
\phi_j(\ajmp{l}) & = & \fjmp{k+2-(j-l)} & \mif 0   <  l   <  j\;, \\
\phi_j(\ajmp{l}) & = & \halt            & \mif l = 0 \lor l > k\;, \\
\phi_j(u)        & = & u
                    & \mif u\; \mathrm{is\;not\;a\;jump\;instruction}\;.
\end{aceqns}
\end{ldispl}%

\PGLD\ is as expressive as \PGA.
Before we make this fully precise, we introduce a useful notation.

Let $\alpha$ is a valuation of variables in $\Ipga$, and let $\alpha^*$
be the unique homomorphic extension of $\alpha$ to terms of \PGA.
Then $\alpha^*(t)$ is independent of $\alpha$ if $t$ is a closed term,
i.e.\ $\alpha^*(t)$ is uniquely determined by $\Ipga$.
Therefore, we write $t^\sIpga$ for $\alpha^*(t)$ if $t$ is a closed
term.

\begin{proposition}
\label{prop-express-PGLD}
For each closed \PGA\ term $t$, there exists a \PGLD\ program $p$ such
that $\textr{t^\sIpga} = \textr{\pgldpga(p)^\sIpga}$.
\end{proposition}
\begin{proof}
In~\cite{BL02a}, a number of functions (called embeddings in that paper)
are defined, whose composition gives, for each closed \PGA\ term $t$, a
\PGLD\ program $p$ such that
$\textr{t^\sIpga} = \textr{\pgldpga(p)^\sIpga}$.
\qed
\end{proof}

Let $p$ be a \PGLD\ program and $P$ be a process.
Then we say that $p$ \emph{produces} $P$ if $\textr{\pgldpga(p)^\sIpga}$
produces $P$.

Below, we will write \PGLDac\ for the version of \PGLD\ in which the
additional assumptions relating to $\BInstr$ mentioned in
Section~\ref{sect-pci} are made.
As a corollary of Theorem~\ref{theorem-completeness} and
Proposition~\ref{prop-express-PGLD}, we have that all regular
processes over $\AAct$ can be produced by \PGLDac\ programs.
\begin{corollary}
\label{corollary-express-PGLDmr}
Assume that CFAR is valid in $\Macp$.
Then, for each process $P$ that is regular over $\AAct$, there exists a
\PGLDac\ program $p$ such that $p$ produces $P$.
\end{corollary}

We switch to the use of Boolean registers now.
First, we describe services that make up Boolean registers.

A Boolean register service accepts the following methods:
\begin{itemize}
\item
a \emph{set to true method} $\setbr{\True}$;
\item
a \emph{set to false method} $\setbr{\False}$;
\item
a \emph{get method} $\getbr$.
\end{itemize}
We write $\Methbr$ for the set
$\set{\setbr{\True},\setbr{\False},\getbr}$.
It is assumed that $\Methbr \subseteq \Meth$.

The methods accepted by Boolean register services can be explained as
follows:
\begin{itemize}
\item
$\setbr{\True}$\,:
the contents of the Boolean register becomes $\True$ and the reply is
$\True$;
\item
$\setbr{\False}$\,:
the contents of the Boolean register becomes $\False$ and the reply is
$\False$;
\item
$\getbr$\,:
nothing changes and the reply is the contents of the Boolean register.
\end{itemize}

Let $s \in \Bool \union \set{\Blocked}$.
Then the \emph{Boolean register service} with initial state $s$, written
$\BR_s$, is the service $\tup{\Bool \union \set{\Blocked},\eff,\eff,s}$,
where the function $\eff$ is defined as follows
($b \in \Bool$):
\begin{ldispl}
\begin{geqns}
\eff(\setbr{\True},b) = \True\;,\;
\\
\eff(\setbr{\False},b) = \False\;,
\\
\eff(\getbr,b) = b\;,
\end{geqns}
\qquad\qquad
\begin{geqns}
\eff(m,b) = \Blocked & \mif m \not\in \Methbr\;,
\\
\eff(m,\Blocked) = \Blocked\;.
\end{geqns}
\end{ldispl}%
Notice that the effect and yield functions of a Boolean register service
are the same.

Let $p$ be a \PGLD\ program and $P$ be a process.
Then we say that $p$ \emph{produces} $P$ \emph{using Boolean registers}
if
$( \ldots
  (\textr{\pgldpga(p)^\sIpga}
    \useop{\br{1}} \BR_\False) \ldots \useop{\br{k}} \BR_\False)$
produces $P$ for some $k \in \Natpos$.

We have that \PGLDac\ programs in which no atomic action from $\AAct$
occurs more than once in an alternative choice instruction can produce
all regular processes over $\AAct$ using Boolean registers.
\begin{theorem}
\label{corollary-express-BR}
Assume that CFAR is valid in $\Macp$.
Then, for each process $P$ that is regular over $\AAct$, there exists a
\PGLDac\ program $p$ in which each atomic action from $\AAct$ occurs no
more than once in an alternative choice instruction such that $p$
produces $P$ using Boolean registers.
\end{theorem}
\begin{proof}
By the proof of Theorem~\ref{theorem-completeness} given in
Section~\ref{sect-expressiveness}, it is sufficient to show that, for
each thread $T$ that is regular over $\BAct$, there exist a \PGLD\
program $p$ in which each basic action from $\BAct$ occurs no more than
once and a $k \in \Natpos$ such that
$( \ldots
  (\textr{\pgldpga(p)^\sIpga}
    \useop{\br{1}} \BR_\False) \ldots \useop{\br{k}} \BR_\False) = T$.

Let $T$ be a thread that is regular over $\BAct$.
We may assume that $T$ is produced by a \PGLD\ program $p'$ of the
following form:
\begin{ldispl}
\ptst{a_1} \conc
 \ajmp{(3 \mul k_1 + 1)} \conc \ajmp{(3 \mul k'_1 + 1)} \conc {}
\\
\quad \vdots
\\
\ptst{a_n} \conc
 \ajmp{(3 \mul k_n + 1)} \conc \ajmp{(3 \mul k'_n + 1)} \conc {}
\\
\ajmp{0} \conc \ajmp{0} \conc \ajmp{0} \conc \ajmp{(3 \mul n + 4)}\;,
\end{ldispl}%
where, for each $i \in [1,n]$, $k_i, k'_i \in [0,n - 1]$
(cf.\ the proof of Proposition~2 from~\cite{PZ06a}).
It is easy to see that the \PGLD\ program $p$ that we are looking for
can be obtained by transforming $p'$: by making use of $n$ Boolean
registers, $p$ can distinguish between different occurrences of the same
basic instruction in $p'$, and in that way simulate $p'$.
\qed
\end{proof}

\section{Conclusions}
\label{sect-concl}

Using the algebraic theory of processes known as \ACP, we have
described two protocols to deal with the phenomenon that, on execution
of an instruction sequence, a stream of instructions to be processed
arises at one place and the processing of that stream of instructions
is handled at another place.
The more complex protocol is directed towards keeping the execution unit
busy.
In this way, we have brought the phenomenon better into the picture and
have ascribed a sense to the term instruction stream which makes clear
that an instruction stream is dynamic by nature, in contradistinction
with an instruction sequence.
We have also discussed some conceivable adaptations of the more complex
protocol.

The description of the protocols start from the behaviours produced by
instruction sequences under execution.
By that we abstract from the instruction sequences which produce those
behaviours.
How instruction streams can be generated efficiently from instruction
sequences is a matter that obviously requires investigations at a less
abstract level.
The investigations in question are an option for future work.

We believe that the more complex protocol described in this paper
provides a setting in which basic techniques aimed at increasing
processor performance, such as pre-fetching and branch prediction, can
be studied at a more abstract level than usual (cf.~\cite{HP03a}).
In particular, we think that the protocol can serve as a starting-point
for the development of a model with which trade-offs encountered in the
design of processor architectures can be clarified.
We consider investigations into this matter an interesting option for
future work.

The fact that process algebra is an area of the study of concurrency
which is considered relevant to computer science, strongly hints that
there are programmed systems whose behaviours are taken for processes
as considered in process algebra.
In that light, we have investigated the connections between programs and
the processes that they produce, starting from the perception of a
program as an instruction sequence.
We have shown that, by apposite choice of basic instructions, all
regular processes can be produced by means of instruction sequences as
considered in \PGA.

We have also made precise what processes are produced by instruction
sequences under execution that make use of services.
The reason for this is that instruction sequences under execution are
regular threads and regular threads that make use of services such as
unbounded counters, unbounded stacks or Turing tapes may produce
non-regular processes.
An option for future work is to characterize the classes of processes
that can be produced by single-pass instruction sequences that make use
of such services.

\bibliographystyle{splncs03}
\bibliography{IS}


\end{document}